\documentclass[11pt,a4paper]{article}
\usepackage[utf8]{inputenc}
\usepackage[T1]{fontenc}
\usepackage[numbers,sort]{natbib}
\usepackage[french,english]{babel}
\usepackage[a4paper,left=2.5cm,right=2.5cm,top=3cm,bottom=3cm,twoside]{geometry}
\usepackage{csquotes}
\usepackage[pdfusetitle]{hyperref}
    \hypersetup{colorlinks=true,linkcolor=[rgb]{0.75,0,0},citecolor=[rgb]{0,0,0.75}}
\usepackage{microtype}
\usepackage{amsmath}
\usepackage{amssymb}
\usepackage{bm}
\usepackage{amsthm}
\usepackage{thmtools}
    \declaretheorem[name=Theorem,numberwithin=section]{theorem} 
    \declaretheorem[name=Lemma,sibling=theorem]{lemma}

\usepackage[table]{xcolor}
\usepackage{multirow}
\usepackage{pbox}
\usepackage{graphicx}
\usepackage{tabularx}
    \newcolumntype{C}{@{}>{\centering\arraybackslash}X@{}}
\usepackage{authblk}

\newcommand{\mc}[2]{\multicolumn{#1}{c|}{#2}}
\newcommand{\graphicsScale}{0.9}
\newcommand{\Figure}{Figure}

\newcommand{\OO}{O} 
\newcommand{\Property}{\mathtt{\Pi}} 
\newcommand{\Multigraph}{\mathtt{Multigraph}} 
\newcommand{\Cycle}{\mathtt{TSP}} 
\newcommand{\Tree}{\mathtt{Tree}} 
\newcommand{\Matching}{\mathtt{Matching}} 
\newcommand{\RedBlue}{\mathtt{RedBlue}} 
\newcommand{\Convex}{\mathtt{Convex}} 

\newcommand{\myC}{C} 
\newcommand{\myT}{T} 
\newcommand{\D}{\mathbf{D}} 
\newcommand{\Choices}{\textit{Choices}} 
\newcommand{\stO}{\emptyset} 
\newcommand{\stR}{\mathtt{R}} 
\newcommand{\stI}{\mathtt{I}} 
\newcommand{\stRI}{\stR\stI} 
\newcommand{\dRI}{\D^{\stRI}} 
\newcommand{\dR}{\D^\stR} 
\newcommand{\dI}{\D^\stI} 
\newcommand{\dO}{\D} 
\newcommand{\mym}{m} 
\newcommand{\nbflips}{k} 
\newcommand{\myn}{n} 
\newcommand{\myk}{k} 
\newcommand{\intFactor}{f} 
\newcommand{\myt}{t} 
\newcommand{\myi}{i} 
\newcommand{\myj}{j} 
\newcommand{\mya}{a} 
\newcommand{\myc}{c} 
\newcommand{\myd}{d} 
\newcommand{\myl}{\ell} 
\newcommand{\myL}{L} 
\newcommand{\myh}{h} 
\newcommand{\myp}{p} 
\newcommand{\myq}{q} 
\newcommand{\myr}{r} 
\newcommand{\myb}{b} 
\newcommand{\myv}{v} 
\newcommand{\myx}{x} 
\newcommand{\myy}{y} 
\newcommand{\myo}{o} 
\newcommand{\ppx}{\myo} 
\newcommand{\mys}{s} 
\newcommand{\myu}{u} 
\newcommand{\flip}{f} 
\newcommand{\myP}{P} 
\newcommand{\myS}{S} 
\newcommand{\inner}{_{\textnormal{in}}} 
\newcommand{\outter}{_{\textnormal{out}}} 
\newcommand{\central}{_{\textnormal{central}}} 
\newcommand{\myQ}{Q} 
\newcommand{\FlipSet}{F} 
\newcommand{\sgt}[2]{#1#2} 
\newcommand{\ray}[2]{#1#2} 
\newcommand{\trgl}[3]{#1#2#3} 
\newcommand{\PotLine}{\Lambda} 
\newcommand{\PotParaLine}{\eta} 
\newcommand{\PotCrossings}{\chi} 
\newcommand{\PotCNC}{\chi} 
\newcommand{\PotCTCNC}{\chi} 
\newcommand{\PotXL}{\Phi} 
\newcommand{\PotDepth}{\delta} 
\newcommand{\PotProduct}{\varpi} 
\newcommand{\crossingDepth}{\delta_\times} 
\newcommand{\PotDotprod}{\rho} 
\DeclareMathOperator{\degree}{deg} 

\newcommand{\myR}{R} 
\newcommand{\myB}{B} 
\newcommand{\lineT}[2]{#1 #2} 
\newcommand{\vect}[2]{\overrightarrow{#1 #2}} 
\newcommand{\nbf}{f} 
\newcommand{\myg}{g} 
\newcommand{\pair}[2]{#1,#2} 
\newcommand{\card}[1]{\left|#1\right|} 
\newcommand{\abs}[1]{\left|#1\right|} 
\newcommand{\floor}[1]{\left\lfloor#1\right\rfloor} 
\newcommand{\ceil}[1]{\left\lceil#1\right\rceil} 
\renewcommand{\vect}[2]{#2 - #1} 
\providecommand{\keywords}[1]{\textbf{\textit{Keywords:}} #1}
\providecommand{\email}[1]{\protect\href{mailto:#1}{#1}}

\graphicspath{{./graphics/}}
\title{Untangling Segments in the Plane\thanks{Preliminary versions of these results appeared in WALCOM 2023~\cite{FGR23}, WALCOM 2024~\cite{FGR24}, and the PhD dissertation~\cite{Riv23}. This work is supported by the French ANR PRC grant ADDS (ANR-19-CE48-0005).}}
\author[1]{Guilherme D. da Fonseca}
\affil[1]{LIS, Aix-Marseille Université, France \href{https://orcid.org/0000-0002-9807-028X}{[ORCID]} \email{guilherme.fonseca@lis-lab.fr}}
\author[2]{Yan Gerard}
\affil[2]{LIMOS, Université Clermont Auvergne, France \href{https://orcid.org/0000-0002-2664-0650}{[ORCID]} \email{yan.gerard@uca.fr}}
\author[3]{Bastien Rivier}
\affil[3]{Brock University, Canada \href{https://orcid.org/0000-0001-5985-2169}{[ORCID]} \email{brivier@brocku.ca}}

\date{}
\begin{document}
\maketitle
\begin{abstract}
A set of $\myn$ segments in the plane may form a Euclidean TSP tour, a tree, or a matching, among others. Optimal TSP tours as well as minimum spanning trees and perfect matchings have no crossing segments, but several heuristics and approximation algorithms may produce solutions with crossings. If two segments cross, then we can reduce the total length with the following flip operation. We remove a pair of crossing segments, and insert a pair of non-crossing segments, while keeping the same vertex degrees.
In this paper, we consider the number of flips performed under different assumptions, using a new unifying framework that applies to tours, trees, matchings, and other types of (multi)graphs.
Within this framework, we prove several new bounds that are sensitive to whether some endpoints are in convex position or not.
\end{abstract}

\section{Introduction}\label{cha:introduction}

In the Euclidean Traveling Salesman Problem (TSP), we are given a set $\myP$ of $\myn$ points in the plane and the goal is to produce a tour of minimum Euclidean length. 
The TSP problem, both in the Euclidean and in the more general graph versions, is one of the most studied NP-hard optimization problems, with several approximation algorithms, as well as powerful heuristics (see for example~\cite{ABCC11,Dav10,GuPu06}).
Multiple PTAS are known for the Euclidean version~\cite{Aro96,Mit99,RaSm98}, in contrast to the general graph version that unlikely admits a PTAS~\cite{ChC19}. It is well known that the optimal solution for the Euclidean TSP is a simple polygon, i.e., a polygon with no crossing segments. In fact, a crossing-free solution is necessary for some applications~\cite{BuKi22}.
However, most approximation algorithms (including Christofides~\cite{Christofides76} and the PTAS~\cite{Aro96,Mit99,RaSm98}), as well as a variety of simple heuristics (nearest neighbor, greedy, and insertion, among others) may produce solutions with pairs of crossing segments.
In practice, these algorithms may be supplemented with a local search phase, in which crossings are removed by iterative modification of the solution.
The objective of this paper is to analyze this local search phase, for TSP as well as for several other problems involving planar segments, using a unifying approach.
Before considering other problems, we provide some definitions and present the previous results on untangling a TSP tour.

\begin{figure}[htb]
 \centering
 \begin{tabularx}{\textwidth}{CCCCC}
        \includegraphics[scale=\graphicsScale,page=4]{graphics/flips.pdf}&%
        \includegraphics[scale=\graphicsScale,page=5]{graphics/flips.pdf}&%
        \includegraphics[scale=\graphicsScale,page=2]{graphics/flips.pdf}&%
        \includegraphics[scale=\graphicsScale,page=1]{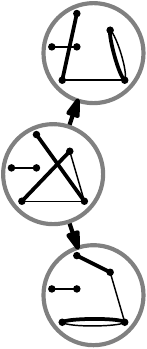}&%
        \includegraphics[scale=\graphicsScale,page=3]{graphics/flips.pdf}\\%
        (a) & (b) & (c) & (d) & (e)
 \end{tabularx}
 \caption{Examples of flips in a (a) TSP tour, (b) tree, (c) monochromatic matching, (d) multigraph, and (e) red-blue matching.}
 \label{fig:flip}
\end{figure}

Given a Euclidean TSP tour, a \emph{flip} is an operation that removes a pair of crossing segments and inserts a new pair of segments preserving a tour (Figure~\ref{fig:flip}(a)).
In order to find a tour without crossing segments starting from an arbitrary tour, it suffices to find a crossing, perform a flip, and repeat until there are no crossings, in a process called \emph{untangle}.
The untangle always terminates as each flip strictly shortens the tour.
Since a flip may create several new crossings (Figure~\ref{fig:flip}(a)), it is not obvious how to bound the number of flips performed to untangle a tour.
Let $\dO_\Cycle(\myn)$ denote the maximum number of flips successively performed on a TSP tour with $\myn$ segments.
An upper bound of $\dO_\Cycle(\myn) = \OO(\myn^3)$ is proved in~\cite{VLe81}, while the best known lower bound is $\dO_\Cycle(\myn) = \Omega(\myn^2)$.
In contrast, if the endpoints $\myP$ satisfy the property of being in convex position, then flips always decrease the number of crossings and tight bounds of $\dO_{\Convex,\Cycle}(\myn) = \Theta(\myn^2)$ are easy to prove.

When the goal is to untangle a TSP tour, we are allowed to \emph{choose} which pair of crossing segments to remove in order to perform fewer flips, which we call \emph{removal choice}.
Let $\dR_\Cycle(\myn)$ denote the number of flips needed to untangle any TSP tour with $\myn$ segments, assuming the best possible removal choice.
If the endpoints $\myP$ are in convex position position, then a tight bound of $\dR_{\Convex,\Cycle}(\myn) = \Theta(\myn)$ has been shown~\cite{OdW07,WCL09}.
However, for points in general position, the best bound known is again $\dR_{\Cycle}(\myn) \leq \dO_{\Cycle}(\myn) = \OO(\myn^3)$, which does not use removal choice, and the only asymptotic lower bound known is the trivial $\dR_{\Cycle}(\myn) = \Omega(\myn)$.

The same flip operation may be applied in other settings. 
More precisely, a \emph{flip} consists of removing a pair of crossing segments $\mys_1,\mys_2$ and inserting a pair of segments $\mys'_1, \mys'_2$ in a way that $\mys_1,\mys'_1,\mys_2,\mys'_2$ forms a cycle and a certain \emph{graph property} is preserved. In the case of \emph{TSP} tours, the property is being a Hamiltonian cycle. Other properties have also been studied, such as spanning \emph{trees} (Figure~\ref{fig:flip}(b)), perfect \emph{matchings} (Figure~\ref{fig:flip}(c)), and \emph{multigraphs} (Figure~\ref{fig:flip}(d)).
Notice that flips preserve the degrees of all vertices and that multiple copies of the same edge may appear when we perform a flip on certain graphs.
Biniaz et al.~\cite{BMS19} showed that if the graph property is being a tree, the points are in convex position, and we perform the correct removal choice, then $ \dR_{\Convex,\Tree}(\myn) = \OO(\myn \log \myn)$ flips suffice to untangle the tree.

Notice that, when the graph property is being a tour or a tree, choosing which pair of crossing edges we remove determines which pair of crossing edges we insert. However, this is not the case for matchings and multigraphs, where we are also allowed to choose which pair of segments to insert among two possibilities, which we call \emph{insertion choice} (Figure~\ref{fig:flip}(c) and (d)).

Bonnet et al.~\cite{BoM16} showed that, using insertion choice, it is possible to untangle a matching using $\dI_\Matching(\myn) = \OO(\myn^2)$ flips.
Let $\sigma$ be the \emph{spread} of $\myP$, that is, the ratio between the maximum and minimum distances among points in $\myP$. Notice that $\sigma = \Omega(\sqrt{\myn})$ but may be arbitrarily large. 
Biniaz et al.~\cite{BMS19} showed an alternative bound of $\dI_{\Property,\Matching}(\myn, \sigma) = \OO(\myn \sigma)$ flips when the endpoints have spread $\sigma$. 
Biniaz et al.~\cite{BMS19} also showed that if both removal and insertion choices are used and $\myP$ is in convex position position, then $\dRI_{\Convex,\Matching}(\myn) = \OO(\myn)$ flips suffice to untangle a matching.

Another version of the problem that has been studied is the \emph{red-blue} version (Figure~\ref{fig:flip}(e)).
In this version, the endpoints $\myP$ are partitioned into $\myn$ \emph{red} points and $\myn$ \emph{blue} points, while the segments must connect points of different colors, forming a red-blue perfect matching.
Notice that insertion choice is not possible in this version.
If the endpoints $\myP$ satisfy the property $\Property$ of all red points being colinear, then using removal choice $\dR_{\Property,\RedBlue}(\myn) = \OO(\myn^2)$ flips suffice~\cite{BMS19,DDFGR22} and the only lower bound known is $\dR_{\Property,\RedBlue}(\myn) = \Omega(\myn)$.
If the endpoints $\myP$ are in convex position, then a tight bound of $\dR_{\Convex,\RedBlue}(\myn) = \Theta(\myn)$ is known~\cite{BMS19}.

\subsection{New Results}

Using removal or insertion choices to obtain shorter flip sequences has not been studied in a systematic way before and opens several new questions, while unifying the solution to multiple  problems. In contrast, previous bounds are usually stated for a single graph property. For example, the seminal TSP bound of $\dO_\Cycle(\myn) = \OO(\myn^3)$ flips~\cite{VLe81} has been rediscovered (with nearly identical proofs) in the context of matchings as $\dO_\Matching(\myn) = \OO(\myn^3)$ after 35 years~\cite{BoM16}.
As another example, the tree bound of $\dR_{\Convex,\Tree}(\myn) = \OO(\myn \log \myn)$ from~\cite{BMS19} uses very specific properties of trees, while our new bound of $\dR_{\Convex}(\myn) = \OO(\myn \log \myn)$ uses a completely different approach that holds for trees, TSP tours, matchings, and arbitrary multigraphs.

Using the framework of choices, we are able to state the results in a more general setting.
Upper bounds for multigraphs that do not use insertion choice apply to all aforementioned problems. In contrast, bounds that use insertion choice are unlikely to generalize to red-blue matchings, TSP tours, or trees, where insertion choice is not available. Still, they generally hold for both monochromatic matchings as well as multigraphs. 

The goal of the paper is to obtain improved bounds in this framework of removal and insertion choices, providing unified proofs to multiple problems. In order to bridge the gaps between points in convex and in general position, we consider configurations where most or all points are in convex position.
Let $\myP = \myC \cup \myT$ where $\myC$ is in convex position and the points of $\myT$ are placed anywhere. Let $\myS$ be a multiset of $\myn$ segments with endpoints $\myP$ and $\myt$ be the sum of the degrees of the points in $\myT$ (and define $\myt = \OO(1)$ if $\myT = \emptyset$). We prove the following results to obtain a crossing-free solution from a set of segments $\myS$. Some bounds are summarized in Table~\ref{tab:results}.

\paragraph{Using no choice (Section~\ref{cha:O}):} We show that a flip sequence has at most $\OO(\myt\myn^2)$ flips. This bound continuously interpolates between the best bounds known for the convex and general cases and holds for TSP tours, matchings, red-blue matching, trees, and multigraphs. In fact, we show that when no choice is used, all lower and upper bounds for matchings can be converted into bounds for the remaining problems (within a constant factor). We also show that any flip sequence has $\OO(\myn^{8/3})$ \emph{distinct} flips.

\paragraph{Using only removal choice (Section~\ref{cha:R}):} If $\abs{\myT} \leq 2$, then $\OO(\myt^2 \myn + \myn  \log \myn)$ flips suffice. These bounds hold for TSP tours, matchings, red-blue matching, trees, and multigraphs.
The $\OO(n \log \myn)$ term is removed for the special cases of TSP tours and red-blue matchings.

\paragraph{Using only insertion choice (Section~\ref{cha:I}):} In the convex case ($\myT=\emptyset$), we show that $\OO(\myn \log \myn)$ flips suffice. If $\myT$ is separated from $\myC$ by two parallel lines, then $\OO(\myt \myn \log \myn)$ flips suffice.

\paragraph{Using both removal and insertion choices (Section~\ref{cha:RI}):} If $\myT$ is separated from $\myC$ by two parallel lines, then $\OO(\myt \myn)$ flips suffice. If $\myT$ is anywhere outside the convex hull of $\myC$ and $\myS$ is a matching, then $\OO(\myt^3 \myn)$ flips suffice.

\begin{table}[htb]
    \small
    \centering
    \begin{tabular}{|l||l|l|l|l|l|l|}\cline{1-7}
Graph: & \mc{2}{TSP, Red-Blue} & \mc{4}{Monochromatic Matching} \\\hline
Choices: & $\emptyset$ & R & $\emptyset$ & R & I & RI\\\hline\hline
Convex & $\bm{\myn^2}$ & $\bm{\myn}$ \cite{BMS19,OdW07,WCL09} & $\bm{\myn^2}$ & \cellcolor{yellow!50}$\myn \log \myn$ (\ref{thm:RUpperConvex}) &  \cellcolor{yellow!50}$\myn \log \myn$ (\ref{thm:convexI}) & $\bm{\myn}$~\cite{BMS19}     \\\hline
$\abs{\myT} = 1$ & \cellcolor{yellow!50}$\bm{\myn^2}$ (\ref{thm:ConvexToGeneralBis}) & \cellcolor{yellow!50}$\bm{\myn}$  (\ref{thm:1InsideOutsideR}) & \cellcolor{yellow!50}$\bm{\myn^2}$ (\ref{thm:ConvexToGeneralBis}) & \cellcolor{yellow!50}$\myn \log \myn$ (\ref{thm:1InsideOutsideR}) & $\myn^2$~\cite{BoM16} & \cellcolor{yellow!50}$\bm{\myn}$ (\ref{thm:1InsideOutsideR}) \\\hline
$\abs{\myT} = 2$ & \cellcolor{yellow!50}$\bm{\myn^2}$ (\ref{thm:ConvexToGeneralBis}) & \cellcolor{yellow!50}$\bm{\myn}$  (\ref{thm:2InsideOutsideR}) & \cellcolor{yellow!50}$\bm{\myn^2}$ (\ref{thm:ConvexToGeneralBis}) & \cellcolor{yellow!50}$\myn \log \myn$  (\ref{thm:2InsideOutsideR}) & $\myn^2$~\cite{BoM16} & \cellcolor{yellow!50}$\bm{\myn}$ (\ref{thm:2InsideOutsideR})\\\hline
$\myT$ separated & \mc{4}{\cellcolor{yellow!50}$\myt\myn^2$ (\ref{thm:ConvexToGeneralBis})} & \cellcolor{yellow!50}$\myt\myn \log \myn$ (\ref{thm:separatedI}) & \cellcolor{yellow!50}$\myt\myn$ (\ref{thm:separatedRI})\\\hline
$\myT$ outside & \mc{4}{\cellcolor{yellow!50}$\myt\myn^2$ (\ref{thm:ConvexToGeneralBis})} & $\myn^2$~\cite{BoM16} & \cellcolor{yellow!50}$\myt^3\myn$ (\ref{thm:nearConvexRI})\\\hline
$\myT$ anywhere & \mc{4}{\cellcolor{yellow!50}$\myt\myn^2$ (\ref{thm:ConvexToGeneralBis})} & \mc{2}{$\myn^2$~\cite{BoM16}}\\\hline    
    \end{tabular}
    \caption{Upper bounds to different versions of the problem with points having $\OO(1)$ degree. The letter R corresponds to removal choice, I to insertion choice, and $\emptyset$ to no choice. New results are highlighted in yellow with the theorem number in parenthesis and tight bounds are bold.}
    \label{tab:results}
\end{table}

\medskip

These bounds improve several previous upper bounds when most points are in convex position with a small number of additional points in general position. Essentially, we obtain linear or near-linear upper bounds when we have removal or insertion choice and quadratic upper bounds when no choice is available. 
Conjecture has been made that the convex upper bounds hold asymptotically for points in general position~\cite{BoM16,DDFGR22}.

In a matching or TSP tour, we have $\myt = \OO(\abs{\myT})$ and $\myn = \OO(\abs{\myP})$, however, in a tree, $\myt$ can be as high as $\OO(\myn)$. 
In a multigraph $\myt$ and $\myn$ can be arbitrarily larger than $\abs{\myT}$ and $\abs{\myP}$. The theorems describe more precise bounds as functions of all these parameters. For simplicity, the introduction only describes the bounds in terms of $\myn$ and $\myt$.

\subsection{Related Reconfiguration Problems}

Combinatorial reconfiguration studies the step-by-step transition from one solution to another for a given combinatorial problem. Many reconfiguration problems are presented in~\cite{Heu13}. We give a brief overview of reconfiguration among $\myn$ line segments using alternative flip operations.

The \emph{2OPT flip} is not restricted to crossing segments. It removes and inserts pairs of segments (the four segments forming a cycle) as long as the total length decreases. In contrast to flips restricted to crossing segments, the number of 2OPT flips performed may be exponential~\cite{ERV14}. 

It is possible to relax the flip definition even further to all operations that replace two segments by two others such that the four segments together form a cycle~\cite{BeI08,BeI17,BBH19,BJ20,EKM13,Wil99}. This definition has also been considered for multigraphs~\cite{Hak62,Hak63,phdJof}. 

Another type of flip consists of removing a single segment and inserting another one. 
Such flips are widely studied for triangulations~\cite{AMP15,HNU99,Law72,LuP15,NiN18,Pil14}.
They have also been considered for non-crossing trees~\cite{ABDK22,BKU24}, paths, and matchings~\cite{ABPS24}. It is possible to reconfigure any two non-crossing paths using $\OO(\myn)$ flips if the points are in convex position~\cite{AIM07,ChWu09,KKR24},  if there is \emph{one} point \emph{inside} the convex hull~\cite{AKLM23}, and if there is \emph{one} point \emph{outside} the convex hull~\cite{BFR24}. It is also possible to reconfigure any two non-crossing paths (using an unknown number of flips) if the endpoints are in two convex layers~\cite{KKR24}.

\subsection{Overview}

In Section~\ref{cha:preliminaries}, we precisely define the notation used throughout the paper and present some key lemmas from previous work. In Section~\ref{cha:O}, we present the bounds without removal or insertion choices as well as several reductions. In Section~\ref{cha:R}, we present the bounds with removal choice. In Section~\ref{cha:I}, we present the bounds with insertion choice. In Section~\ref{cha:RI}, we present the bounds with both removal and insertion choices. Concluding remarks and open problems are presented in Section~\ref{cha:conclusion}.

\section{Preliminaries}\label{cha:preliminaries}
This section contains a collection of definitions, and lemmas used throughout the paper.

\paragraph{Asymptotic bounds incorrectly fading to zero.} Our bounds often have terms like $\OO(\myt\myn)$ and $\OO(\myn \log \card{\myC})$ that would incorrectly become $0$ if $\myt$ or $\log \card{\myC}$ is $0$. In order to avoid this problem, factors in the $\OO$ notation should be made at least 1. For example, the aforementioned bounds should be respectively interpreted as $\OO((1+\myt) \myn)$ and $\OO(\myn \log (2+\card{\myC}))$.

\paragraph{General position.} Throughout this paper, we assume general position of the endpoints of the segments we consider. In particular, we exclude three colinear endpoints, in which case a flip and a crossing are not well defined. In some proofs, we assume that all $y$-coordinates are distinct, but it is straightforward to remove this assumption. We also assume that the two endpoints of a segment are always distinct.

\subsection{General Definitions}
In the following, we summarize important definitions used throughout this paper.

\paragraph{Segment types.} Given two (possibly equal) sets $\myP_1,\myP_2$ of endpoints, we say that a segment is a \emph{$\myP_1\myP_2$-segment} if one of its endpoints is in $\myP_1$ and the other is in $\myP_2$.

\paragraph{Crossings.} We say that two segments \emph{cross} if they intersect at a single point which is not an endpoint of either segment. A pair of crossing segments is called a \emph{crossing}. We also say that a line and a segment \emph{cross} if they intersect at a single point which is not an endpoint of the segment.

If $\myh$ is a segment or a line (respectively a set of two parallel lines) we say that $\myh$ \emph{separates} a set of points $\myP$ if $\myP$ can be partitioned into two non-empty sets $\myP_1,\myP_2$ such that every segment $\sgt{\myp_1}{\myp_2}$ with $\myp_1 \in \myP_1, \myp_2 \in \myP_2$ crosses $\myh$ (respectively crosses at least one of the lines in $\myh$).

\subsection{Properties}

We freely interpret a multiset $\myS$ of $\myn$ segments with endpoints $\myP$ in the plane as the multigraph $(\myP,\myS)$ where the set of vertices is the set of endpoints and where the multiset of edges is the multiset of segments. In some cases the endpoints $\myP$ have a color (typically red or blue).

We consider the following two types of properties.

\begin{enumerate}
    \item The \emph{point properties} are the properties about the position of the endpoints. A notable point property is that the endpoints are in convex position. Throughout, we refer to this property as $\Convex$.
    \item The \emph{graph properties} are the class of graphs before and after a flip. Example: the segments form a matching.
\end{enumerate}

The \textbf{graph properties} $\Gamma$ considered is this paper are the following:
\begin{itemize}
    \item $\Cycle$: Being a Hamiltonian cycle.
    \item $\Matching$: Being a perfect matching.
    \item $\RedBlue$: Being a perfect matching where each segment connects a red endpoint to a blue endpoint.
    \item $\Tree$: Being a spanning tree.
    \item $\Multigraph$: Being a multigraph. Since this property imposes no restriction on the segments, it may not be explicitly written on the bounds.
\end{itemize}

\subsection{Flips and Flip Graphs}

Given a graph property $\Gamma$ and a multiset of segments $S$ satisfying $\Gamma$, a \emph{flip} $\flip$ is the operation of removing two crossing segments $\mys_1,\mys_2 \in S$ and inserting two non-crossing segments $\mys_1',\mys_2'$ with the same four endpoints that preserve the property $\Gamma$.

Given a set of endpoints $P$, an integer $n$, and a graph property $\Gamma$, the \emph{flip graph} $F(P,n,\Gamma)$ is defined as the following directed acyclic simple graph. The vertices of the flip graph are the multisets of $n$ segments with endpoints in $P$ that satisfy the property $\Gamma$. There is a directed arc from $S_1$ to $S_2$ if there exists a flip $\flip$ such that $S_2 = \flip(S_1)$. Figure~\ref{fig:reconf}(a) and~\ref{fig:strategy}(a) display examples of flip graphs in the TSP and matching versions, respectively. 

\begin{figure}[htb]
    \centering\hspace*{\stretch{1}}%
    \pbox[b]{\textwidth}{\centering\includegraphics[scale=\graphicsScale,page=1]{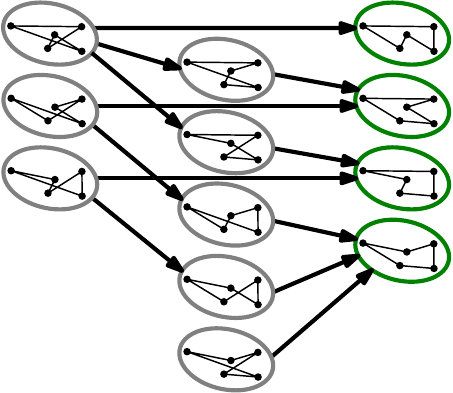}\newline(a)}\hspace*{\stretch{2}}%
    \pbox[b]{\textwidth}{\centering\includegraphics[scale=\graphicsScale,page=2]{sequence}\newline(b)}\hspace*{\stretch{2}}%
 \caption{The (a) flip graph and a (b) removal strategy of a five-point set in the TSP version.
 }
 \label{fig:reconf}
\end{figure}

Given a flip graph $F$, a \emph{strategy} is a spanning subgraph of $F$ with the same set of sink vertices (a sink vertex corresponds to a crossing-free set of segments).
A \emph{removal strategy} is a strategy $R$ such that each vertex $S$ of $F$ has an associated crossing pair of segments $s_1,s_2 \in S$ and $R$ only contains the arcs from $S$ that remove $s_1,s_2$. Figures~\ref{fig:reconf}(b) and~\ref{fig:strategy}(b) display examples of removal strategies in the TSP and matching versions, respectively. 

The following strategies with insertion choice are only possible if the graph property $\Gamma$ is such that for every crossing pair of segments removed, there are two possible pairs of segment that may be added.
A \emph{strategy for both removal and insertion choices} is a strategy $RI$ such that each vertex of $RI$ that is not a sink in $F$ has out-degree exactly one. 
An \emph{insertion strategy} is a strategy $I$ obtained from $F$ by erasing the following arcs. For each pair of arcs $f,f'$ coming from the same vertex $S$ and with the same pair of crossing segments $s_1,s_2 \in S$ removed, exactly one of the two arcs $f,f'$ is erased and the other is kept in $I$. Figure~\ref{fig:strategy}(c) and~\ref{fig:strategy}(d)  display examples of a removal strategy and a strategy for both choices, respectively, in the matching version.

\begin{figure}[p]
    \centering
    \hspace*{\stretch{1}}%
    \pbox[b]{\textwidth}{\centering\includegraphics[scale=\graphicsScale,page=1]{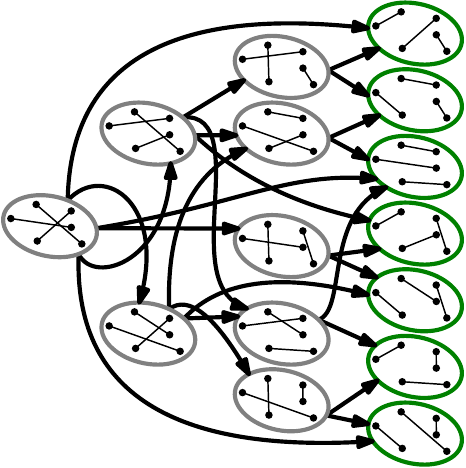}\newline(a)}\hspace*{\stretch{2}}%
    \pbox[b]{\textwidth}{\centering\includegraphics[scale=\graphicsScale,page=2]{strategy}\newline(b)}\hspace*{\stretch{2}}\\
    \medskip
    \hspace*{\stretch{1}}%
    \pbox[b]{\textwidth}{\centering\includegraphics[scale=\graphicsScale,page=3]{strategy}\newline(c)}\hspace*{\stretch{2}}%
    \pbox[b]{\textwidth}{\centering\includegraphics[scale=\graphicsScale,page=4]{strategy}\newline(d)}\hspace*{\stretch{2}}\\    
 \caption{(a) The flip graph $F$ of a six-point set in the matching version. (b) Removal, (c) insertion, and (d) both removal and insertion strategies for the flip graph $F$. Pairs of edges corresponding to the same crossing are grouped.
 }
 \label{fig:strategy}
\end{figure}

Let $\mathcal{R}(F),\mathcal{I}(F), \mathcal{RI}(F)$ respectively denote the collection of all strategies for the flip graph $F$ that are removal strategies, insertion strategies, and strategies for both removal and insertion choices.

A \emph{flip sequence} is a path in the flip graph (or a strategy, which is a spanning subgraph of the flip graph) and an \emph{untangle sequence} is a flip sequence that ends in a sink of the flip graph. The \emph{length} of a flip sequence is the number of flips (the arcs of the flip graph) it contains. Given a graph $F$, let $L(F)$ denote the length of the longest directed path in $F$. Notice that $L(F)=3$ for the flip graph $F$ in Figure~\ref{fig:strategy}(a), $L(R)=2$ for the removal strategy $R$ in Figure~\ref{fig:strategy}(b), $L(I)=2$ for the insertion strategy $I$ in Figure~\ref{fig:strategy}(c), and $L(RI)=1$ for the removal and insertion strategy $RI$ in Figure~\ref{fig:strategy}(d).

We are now ready to formally define the following four parameters, which are the central object of this paper.
\begin{alignat*}{4}
    \dO_{\Property,\Gamma}(\myn) \quad &&= \quad \max_{P \text{ satisfying } \Property} \quad &&\;&& L(F(P,n,\Gamma)) \\
    \dR_{\Property,\Gamma}(\myn) \quad &&= \quad \max_{P \text{ satisfying } \Property} \quad && \min_{R \in \mathcal{R}(F(P,n,\Gamma))} && L(R) \\
    \dI_{\Property,\Gamma}(\myn) \quad &&= \quad \max_{P \text{ satisfying } \Property} \quad && \min_{I \in \mathcal{I}(F(P,n,\Gamma))} && L(I) \\
    \dRI_{\Property,\Gamma}(\myn) \quad &&= \quad \max_{P \text{ satisfying } \Property} \quad && \min_{RI \in \mathcal{RI}(F(P,n,\Gamma))} && L(RI) 
\end{alignat*}
To simplify the notation, we omit $\Property$ when it is the most general point property of being a set of points in general position, and we omit $\Gamma$ when it is the most general graph property of being any multigraph.

It is convenient to imagine that two clever players, let us call them \emph{the oracle} and \emph{the adversary}, are playing the following game using the flip graph as a board. The oracle aims at minimizing the number of flips, while the adversary aims at maximizing it. Initially, the adversary chooses the starting vertex in the flip graph. The game ends when a crossing-free vertex is reached. The four definitions corresponds to the number of flips with different choices for the oracle (specified in the exponent).
The number of flips is $\dO$ if the adversary performs all the choices, it is $\dR$ if the oracle performs all the removal choices while the adversary performs all the insertion choices, it is $\dI$ if the oracle performs all the insertion choices while the adversary performs all the removal choices, it is $\dRI$ if the oracle performs all the removal and all the insertion choices (the adversary only gets to choose the starting vertex).

\subsection{Reductions}

In this section, we state some trivial inequalities between the different versions of $\D$. First, as we have fewer choices, the value of $\D$ can only increase.

\begin{lemma}
    \label{lem:reductionChoice}
    The following inequalities hold for any non-negative integer $\myn$,  point property $\Property$, and graph property $\Gamma$.
    \[
    \dRI_{\Property,\Gamma}(\myn) \leq
        \dR_{\Property,\Gamma}(\myn) 
    \leq \dO_{\Property,\Gamma}(\myn)
    \]
    \[
    \dRI_{\Property,\Gamma}(\myn) \leq
        \dI_{\Property,\Gamma}(\myn)
    \leq \dO_{\Property,\Gamma}(\myn)
    \]
\end{lemma}

Also, as the point set becomes less constrained, the value of $\D$ can only increase. Furthermore, if the graph property is less constrained, it is easy to show that the value of $\D$ can only increase using the fact that insertion strategies are only possible when both insertion choices preserve the graph property at every flip.

\begin{lemma}
    \label{lem:reductionProperty}
    The following inequality holds for any non-negative integer $\myn$, two point properties properties $\Property,\Property'$ such that $\Property$ implies $\Property'$,  two graph properties properties $\Gamma,\Gamma'$ such that $\Gamma$ implies $\Gamma'$, and for any $\Choices \in  \{\stO, \stR, \stI, \stRI\}$ that are possible for $\Gamma$.
    \[ \D^\Choices_{\Property,\Gamma}(\myn) \leq \D^\Choices_{\Property',\Gamma'}(\myn) \]
\end{lemma}

Finally, we present some inequalities relating matchings and red-blue matchings.

\begin{lemma}
    \label{lem:reductionTransfer}
    The following inequalities hold for any non-negative integer $\myn$, and for any point property $\Property$.
    \begin{align*}
        \dRI_{\Property,\Matching}(\myn) & \leq \dR_{\Property,\RedBlue}(\myn) \\
        \dI_{\Property,\Matching}(\myn) & \leq \dO_{\Property,\RedBlue}(\myn)
    \end{align*}
\end{lemma}

\subsection{Splitting Lemma}

The notion of splitting first appears in~\cite{BoM16}. The idea is that some subsets of segments may be untangled independently. We start with a definition. 

Given a multiset $S$ of segments, a \emph{splitting partition} is a partition of $S$ into $k > 1$ subsets $S_1,\ldots,S_k$ such that the following condition holds.
Let $P_i$ be the set of endpoints of $S_i$ and $\binom{P_i}{2}$ be the set of segments with endpoints in $P_i$. The condition is that there is no pair of crossing segments $s_i \in \binom{P_i}{2}$ and $s_j \in \binom{P_j}{2}$ for $i \neq j$.

\begin{lemma}[\cite{BoM16,FGR24}]
\label{lem:splitting}
Consider a multiset of segments $S$ and a splitting partition $S_1,\ldots,S_k$. The subsets $S_1,\ldots,S_k$ may be untangled independently in any order to obtain an untangle sequence of $S$ and every untangle sequence of $S$ may be obtained in such a way.
\end{lemma}

If a singleton $\{\mys\}$ is in a splitting partition, then we say that the segment $\mys$ is \emph{uncrossable}.

\section{Untangling with No Choice}
\label{cha:O}
In this section, we study $\dO$, the length of the longest untangle sequence in the flip graph.
First, we prove reductions between several versions of $\dO$.
Then, we prove an upper bound parameterized by $\myt$ (the sum of the degrees of the points which are not in convex position).
Finally, we prove a sub-cubic upper bound on the number of \emph{distinct} flips in an untangle sequence. This number is \emph{not} the length of an untangle sequence where the same flip may appear several times (up to a linear number of times) and is counted with its multiplicity.

\subsection{Reductions}
\label{sec:Oreductions}

In this section, we provide a series of inequalities relating the different versions of $\dO$. In particular, we show that $\dO_\Matching$, $\dO_\RedBlue$, $\dO_\Cycle$, $\dO_\Tree$, and $\dO_\Multigraph$ have the same asymptotic behavior. 

We say that a point property $\Property$ is \emph{replicable} if for any point set $P$ that satisfies $\Property$, any function $k : P \to \mathbb{N}$, and any $\epsilon > 0$, there exists a point set $P'$ that satisfies $\Property$ replacing each point $ p \in P $ by $ k ( p ) $ points located within distance at most $\epsilon$ from $p$.
Notice that point properties of being in general and convex position are replicable, while the point property of having at most one point inside the convex hull is not.

\begin{theorem}\label{thm:Oreductions}
  For all positive integer $\myn$ and for all replicable point property $\Property$, we have the following relations:
  \begin{align}
     \dO_{\Property, \Matching}(\myn) & = \dO_{\Property, \Multigraph}(\myn), \label{eq:0}\\
     2\,\dO_{\Property, \Matching}(\myn) & \leq \dO_{\Property, \RedBlue}(2\myn)~ \leq \dO_{\Property, \Matching}(2\myn), \label{eq:1}\\
     2\,\dO_{\Property, \RedBlue}(\myn) & \leq \dO_{\Property, \Cycle}(3\myn) \leq \dO_{\Property, \Matching}(3\myn), \label{eq:2}\\
     2\,\dO_{\Property, \RedBlue}(\myn) & \leq \dO_{\Property, \Tree}(3\myn) \leq \dO_{\Property, \Matching}(3\myn). \label{eq:3}
  \end{align}
\end{theorem}
\begin{proof}

Equality~\eqref{eq:0} can be rewritten $\dO_{\Property, \Multigraph}(\myn) \leq \dO_{\Property, \Matching}(\myn) \leq  \dO_{\Property, \Multigraph}(\myn)$. Hence, we have to prove eight inequalities.
The right-side inequalities are trivial, since the left-side property is stronger than the right-side property (using the equality~\eqref{eq:0} for inequalities~\eqref{eq:2} and~\eqref{eq:3}).

The proofs of the remaining inequalities follow the same structure: given a flip sequence of the left-side version, we build a flip sequence of the right-side version, having similar length and number of points.

\paragraph{Proving $\dO_{\Property, \Multigraph}(\myn) \leq \dO_{\Property, \Matching}(\myn)$~\eqref{eq:0}.}
We prove the left inequality of~\eqref{eq:0}. 
A point of degree $\degree$ larger than $1$ can be replicated as $\degree$ points that are arbitrarily close to each other in order to produce a matching of $2\myn$ points.
This replication preserves the crossing pairs of segments, possibly creating new crossings (Figure~\ref{fig:simulateFlip}(a)). Thus, for any flip sequence in the $\Multigraph$ version, there exists a flip sequence in the $\Matching$ version of equal length, yielding $\dO_{\Property, \Multigraph}(\myn) \leq \dO_{\Property, \Matching}(\myn)$.

\begin{figure}[!ht]
    \centering\hspace*{\stretch{1}}%
    \pbox[b]{\textwidth}{\centering\includegraphics[scale=\graphicsScale,page=1]{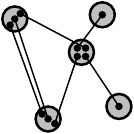}\newline(a)}\hspace*{\stretch{2}}%
    \pbox[b]{\textwidth}{\centering\includegraphics[scale=\graphicsScale,page=1]{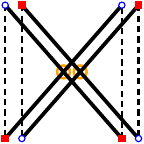}\newline(b)}\hspace*{\stretch{2}}%
    \pbox[b]{\textwidth}{\centering\includegraphics[scale=\graphicsScale,page=1]{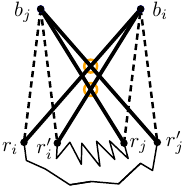}\newline(c)}\hspace*{\stretch{2}}%
    \pbox[b]{\textwidth}{\centering\includegraphics[scale=\graphicsScale,page=2]{MMTSP}\newline(d)}\hspace*{\stretch{1}}%
    \caption{(a) A crossing-free multigraph transformed into a matching by replacing multi-degree points by clusters of points. (b) Two flips in the $\RedBlue$ version simulating one flip in the $\Matching$ version. (c) Two flips in the $\Cycle$ version simulating one flip in the $\RedBlue$ version. (d) Two flips in the $\Tree$ version simulating one flip in the $\RedBlue$ version.}%
  \label{fig:simulateFlip}%
\end{figure}

\paragraph{Proving $2\,\dO_{\Property, \Matching}(\myn) \leq \dO_{\Property, \RedBlue}(2\myn)$~\eqref{eq:1}.}
The left inequality of~\eqref{eq:1} is obtained by duplicating the monochromatic points of the matching $\myS$ into two arbitrarily close points, one red and the other blue. Then each segment of $\myS$ is also duplicated into two red-blue segments. We obtain a red-blue matching $\myS'$ with $2\myn$ segments. 
A crossing in $\myS$ corresponds to four crossings in $\myS'$. Flipping this crossing in $\myS$ amounts to choose which of the two possible pairs of segments replaces the crossing pair. It is simulated by flipping the two crossings in $\myS'$ such that the resulting pair of double segments corresponds to the resulting pair of segments of the initial flip. These two crossings always exist and it is always possible to flip them one after the other as they involve disjoint pairs of segments.
\Figure~\ref{fig:simulateFlip}(b) shows this construction.
A sequence of $\nbflips$ flips on $\myS$ provides a sequence of $2\nbflips$ flips on $\myS'$.
Hence, $2 \dO_{\Property, \Matching}(\myn) \leq \dO_{\Property, \RedBlue}(2\myn)$.

\paragraph{Proving $2\,\dO_{\Property, \RedBlue}(\myn) \leq \dO_{\Property, \Cycle}(3\myn)$~\eqref{eq:2}.}
To prove the left inequality of~\eqref{eq:2}, we start from a red-blue matching $\myS$ with $2\myn$ points and $\myn$ segments and build a cycle $\myS'$ with $3\myn$ points and $3\myn$ segments. We then show that the flip sequence of length $\nbflips$ on $\myS$ provides a flip sequence of length $2\nbflips$ on $\myS'$. We build $\myS'$ in the following way. Given a red-blue segment $\sgt{\myr}{\myb} \in \myS$, the red point $\myr$ is duplicated in two arbitrarily close points $\myr$ and $\myr'$ which are adjacent to $\myb$ in $\myS'$. We still need to connect the points $\myr$ and $\myr'$ in order to obtain a cycle $\myS'$. 
We define $\myS'$ as the cycle
$\myr_1,\myb_1,\myr'_1,\ldots, \myr_\myi,\myb_\myi,\myr'_\myi,\ldots,\myr_\myn,\myb_\myn,\myr'_\myn\ldots$
where $\myr_\myi$ is matched to $\myb_\myi$ in $\myS$ (\Figure~\ref{fig:simulateFlip}(c)).

We now show that a flip sequence of $\myS$ with length $\nbflips$ provides a flip sequence of $\myS'$ with length $2\nbflips$. For a flip on $\myS$ removing the pair of segments $\pair{\sgt{\myr_\myi}{\myb_\myi}}{\sgt{\myr_\myj}{\myb_\myj}}$ and inserting the pair of segments $\pair{\sgt{\myr_\myi}{\myb_\myj}}{\sgt{\myr_\myj}{\myb_\myi}}$, we perform the following two successive flips on $\myS'$.
\begin{itemize}
    \item The first flip removes $\pair{\sgt{\myr_\myi}{\myb_\myi}}{\sgt{\myr'_\myj}{\myb_\myj}}$ and inserts $\pair{\sgt{\myr_\myi}{\myb_\myj}}{\sgt{\myr'_\myj}{\myb_\myi}}$.
    \item The second flip removes $\pair{\sgt{\myr'_\myi}{\myb_\myi}}{\sgt{\myr_\myj}{\myb_\myj}}$ and inserts $\pair{\sgt{\myr'_\myi}{\myb_\myj}}{\sgt{\myr_\myj}{\myb_\myi}}$.
\end{itemize}

The cycle then becomes
$
\myr_1,\myb_1,\myr'_1,\ldots, \myr_\myi,\myb_\myj,\myr'_\myi,\ldots, \myr_\myj,\myb_\myi,\myr'_\myj,\ldots,\myr_\myn,\myb_\myn,\myr'_\myn,\ldots
$
on which we can apply the next flips in the same way. Hence, $2 \dO_{\Property, \Matching}(\myn) \leq \dO_{\Property, \Cycle}(3\myn)$.

\paragraph{Proving $2\,\dO_{\Property, \RedBlue}(\myn) \leq \dO_{\Property, \Tree}(3\myn)$~\eqref{eq:3}.}
The proof of the left inequality of~\eqref{eq:3} follows the exact same construction as in the proof of the left inequality of~\eqref{eq:2}. The only difference is that, in order for $\myS'$ to form a tree and not a cycle, we omit the segment $\sgt{\myr'_\myn}{\myr_1}$, yielding a polygonal line which is a tree (\Figure~\ref{fig:simulateFlip}(d)).
\end{proof}

\subsection{Upper Bound for Near Convex Position}
\label{sec:ConvexToGeneral}

In this section, we bridge the gap between the $\OO(\myn^2)$ bound on the length of untangle sequences for a set $\myP$ of points in convex position and the $\OO(\myn^3)$ bound for $\myP$ in general position.
We prove the following theorem in the $\Matching$ version; the translation to the multigraph version, follows from the reductions in Theorem~\ref{thm:Oreductions} and then other graph properties follow from Lemma~\ref{lem:reductionProperty}.

\begin{theorem}\label{thm:ConvexToGeneralBis}
    Consider a multiset $\myS$ of $\myn$ segments with endpoints $\myP$ partitioned into $\myP = \myC \cup \myT$ where $\myC$ is in convex position. Let $\myt$ be the sum of the degrees of the points in $\myT$. Any untangle sequence of $\myS$ has length at most 
    \[\dO(\myn,\myt) = \OO(\myt \myn^2).\]
\end{theorem}
\begin{proof}
The proof strategy is to combine the potential $\PotCrossings$ used in~\cite{BMS19} with the potential $\PotLine_\myL$ used in~\cite{LS80}.
Given a matching $\myS$, the potential $\PotCrossings(\myS)$ is defined as the number of crossing pairs of segments in $\myS$. Since there are $\myn$ segments in $\myS$, $\PotCrossings(\myS) \leq \binom{\myn}{2} = \OO(\myn^2)$.
Unfortunately, with points in non-convex position, a flip $\flip$ might \emph{increase} (or leave unchanged) $\PotCrossings$, i.e. $\PotCrossings(\flip(\myS)) \geq \PotCrossings(\myS)$.

The potential $\PotLine_\myL$ is derived from the line potential introduced in~\cite{VLe81} but instead of using the set of all the $\OO(\myn^2)$ lines through two points of $\myP$, we use a subset of $\OO(\myt\myn)$ lines in order to take into account that only $\myt$ points are in non-convex position. 
More precisely, let the potential $\PotLine_\myl(\myS)$ of a line $\myl$ be the number of segments of $\myS$ crossing $\myl$. Note that $\PotLine_\myl(\myS) \leq \myn$. The potential $\PotLine_\myL(\myS)$ is then defined as follows: $\PotLine_\myL(\myS) = \sum_{\myl \in \myL} \PotLine_\myl(\myS)$.

We now define the set of lines $\myL$ as the union of $\myL_1$ and $\myL_2$, defined hereafter. Let $\myC$ be the subset containing the $2\myn-\myt$ points of $\myP$ which are in convex position. Let $\myL_1$ be the set of the $\OO(\myt\myn)$ lines through two points of $\myP$, at least one of which is not in $\myC$. Let $\myL_2$ be the set of the $\OO(\myn)$ lines through two points of $\myC$ which are consecutive on the convex hull boundary of $\myC$.

Let the potential $\PotXL(\myS) = \PotCrossings(\myS) + \PotLine_\myL(\myS)$.
We have the following bounds: $0 \leq \PotXL(\myS) \leq \OO(\myt\myn^2)$.
To complete the proof of Theorem~\ref{thm:ConvexToGeneralBis}, we show that any flip decreases $\PotXL$ by at least $1$.

We consider an arbitrary flip $\flip$ removing the pair of segments $\pair{\sgt{\myp_1}{\myp_3}}{\sgt{\myp_2}{\myp_4}}$ and inserting the pair of segments $\pair{\sgt{\myp_1}{\myp_4}}{\sgt{\myp_2}{\myp_3}}$.
Let $\ppx$ be the point of intersection of $\sgt{\myp_1}{\myp_3}$ and $\sgt{\myp_2}{\myp_4}$.
It is shown in~\cite{VLe81} that $\flip$ never increases the potential $\PotLine_\myl$ of a line $\myl$. More precisely, we have the following three cases:
\begin{itemize}
    \item The potential $\PotLine_\myl$ decreases by $1$ if the line $\myl$ separates the final segments $\sgt{\myp_1}{\myp_4}$ and $\sgt{\myp_2}{\myp_3}$ and exactly one of the four flipped points belongs to $\myl$. We call these lines \emph{$\flip$-critical} (\Figure~\ref{fig:criticalPotentialDrop}(a)).
    \item The potential $\PotLine_\myl$ decreases by $2$ if the line $\myl$ strictly separates the final segments $\sgt{\myp_1}{\myp_4}$ and $\sgt{\myp_2}{\myp_3}$. We call these lines \emph{$\flip$-dropping} (\Figure~\ref{fig:criticalPotentialDrop}(b)). 
    \item The potential $\PotLine_\myl$ remains stable in the remaining cases.
\end{itemize}
Notice that, if a point $\myq$ lies in the triangle $\trgl{\myp_1}{\ppx}{\myp_4}$, then the two lines $\lineT{\myq}{\myp_1}$ and $\lineT{\myq}{\myp_4}$ are $\flip$-critical (\Figure~\ref{fig:criticalPotentialDrop}(a)). 

\begin{figure}[!ht]
    \centering\hspace*{\stretch{1}}%
    \pbox[b]{\textwidth}{\centering\includegraphics[scale=\graphicsScale,page=1]{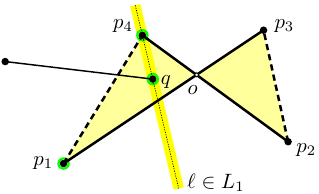}\newline(a)}\hspace*{\stretch{2}}%
    \pbox[b]{\textwidth}{\centering\includegraphics[scale=\graphicsScale,page=2]{criticalPotentialDrop}\newline(b)}\hspace*{\stretch{1}}%
  \caption{(a) An $\flip$-critical line $\myl$ for a flip $\flip$ removing $\pair{\sgt{\myp_1}{\myp_3}}{\sgt{\myp_2}{\myp_4}}$ and inserting $\pair{\sgt{\myp_1}{\myp_4}}{\sgt{\myp_2}{\myp_3}}$. This situation corresponds to case (2a) with $\myl \in \myL_1$. (b) An $\flip$-dropping line $\myl$. This situation corresponds to case (2b) with $\myl \in \myL_2$.} 
  \label{fig:criticalPotentialDrop}%
\end{figure}

To prove that $\PotXL$ decreases, we have the following two cases.

\textbf{Case 1.} If $\PotCrossings$ decreases, as the other term $\PotLine_\myL$ does not increase, then their sum $\PotXL$ decreases as desired. 

\textbf{Case 2.} If not, then $\PotCrossings$ increases by an integer $\myk$ with $0 \leq k \leq \myn-1$, and we know that there are $\myk+1$ new crossings after the flip $\flip$. Each new crossing involves a distinct segment with one endpoint, say $\myq_\myi$ ($0 \leq i \leq \myk$), inside the non-simple polygon $\myp_1,\myp_4,\myp_2,\myp_3$ (\Figure~\ref{fig:criticalPotentialDrop}). 
Next, we show that each point $\myq \in \{\myq_0,\ldots,\myq_\myk\}$ maps to a distinct line in $\myL$ which is either $\flip$-dropping or $\flip$-critical, thus proving that the potential $\PotLine_\myL$ decreases by at least $\myk+1$.

We assume without loss of generality that $\myq$ lies in the triangle $\myp_1\ppx\myp_4$.
We consider the two following cases.

\textbf{Case 2a.} If at least one among the points $\myq, \myp_1, \myp_4$ is not in $\myC$, then either $\myq\myp_1$ or $\myq\myp_4$ is an $\flip$-critical line $\myl \in \myL_1$ (\Figure~\ref{fig:criticalPotentialDrop}(a)).

\textbf{Case 2b.} If not, then $\myq, \myp_1, \myp_4$ are all in $\myC$, and the two lines through $\myq$ in $\myL_2$ are both either $\flip$-dropping (the line $\myl$ in \Figure~\ref{fig:criticalPotentialDrop}(b)) or $\flip$-critical (the line $\myq\myp_4$ in \Figure~\ref{fig:criticalPotentialDrop}(b)).
Consequently, there are more lines $\myl \in \myL_2$ that are either $\flip$-dropping or $\flip$-critical than there are such points $\myq \in \myC$ in the triangle $\trgl{\myp_1}{\ppx}{\myp_4}$, and the theorem follows.
\end{proof}

\subsection{Upper Bound without Multiplicity}
\label{sec:distinct}

In this section, we prove the following theorem. We recall that two flips are \emph{distinct} if the set of the two removed and the two inserted segments of one flip is distinct from the set of the other flip.

\begin{theorem}
    \label{thm:distinct}
    Consider a multiset $\myS$ of $\myn$ segments with endpoints $\myP$.  
    Any untangle sequence of $\myS$ has $\OO(\myn^{{8}/{3}})$ \emph{distinct} flips.
\end{theorem}

The proof of Theorem~\ref{thm:distinct} is based on a balancing argument from~\cite{ELSS73} and is decomposed into two lemmas that consider a flip $\flip$ and two sets of segments $\myS$ and $\myS' = \flip(\myS)$. Similarly to~\cite{VLe81}, let $\myL$ be the set of lines defined by all pairs of points in $\binom{\myP}{2}$. For a line $\myl \in \myL$, let $\PotLine_\myl(\myS)$ be the number of segments of $\myS$ crossed by $\myl$ and $\PotLine_\myL(\myS)=\sum_{\myl \in \myL} \PotLine_\myl(\myS)$. Notice that $\PotLine_\myL(\myS) - \PotLine_\myL(\myS')$ depends only on the flip $\flip$. The following lemma follows immediately from the fact that $\PotLine_\myL(\myS)$ takes integer values between $0$ and $\OO(\myn^3)$.

\begin{lemma}
  \label{lem:bigDrops}
  For any integer $\myk$, the number of flips $\flip$ in a flip sequence with $\PotLine_\myL(\myS) - \PotLine_\myL(\myS') \geq \myk$ is $\OO(\myn^3/\myk)$.
\end{lemma}

Lemma~\ref{lem:bigDrops} bounds the number of flips (distinct or not) that produce a large potential drop in a flip sequence. Next, we bound the number of distinct flips that produce a small potential drop. The bound considers all possible flips on a fixed set of points and does not depend on a particular flip sequence.

\begin{lemma}
  \label{lem:smallDrops}
  For any integer $\myk$, the number of \emph{distinct} flips $\flip$ with $\PotLine_\myL(\myS) - \PotLine_\myL(\myS') < \myk$ is $\OO(\myn^2\myk^2)$.
\end{lemma}
\begin{proof}
Let $\FlipSet$ be the set of flips with $\PotLine_\myL(\myS) - \PotLine_\myL(\myS') < \myk$ where $\myS'=f(\myS)$. We need to show that $\card{\FlipSet} = \OO(\myn^2\myk^2)$. Consider a flip $\flip \in \FlipSet$ removing the pair of segments $\pair{\sgt{\myp_1}{\myp_3}}{\sgt{\myp_2}{\myp_4}}$ and inserting the pair of segments $\pair{\sgt{\myp_1}{\myp_4}}{\sgt{\myp_2}{\myp_3}}$. Next, we show that there are at most $4\myk^2$ such flips with a fixed final segment $\myp_1\myp_4$. Since there are $\OO(\myn^2)$ possible values for $\myp_1\myp_4$, the lemma follows. We show only that there are at most $2\myk$ possible values for $\myp_3$. The proof that there are at most $2\myk$ possible values for $\myp_2$ is analogous.

\begin{figure}[!ht]
  \centering%
  \includegraphics[scale=\graphicsScale]{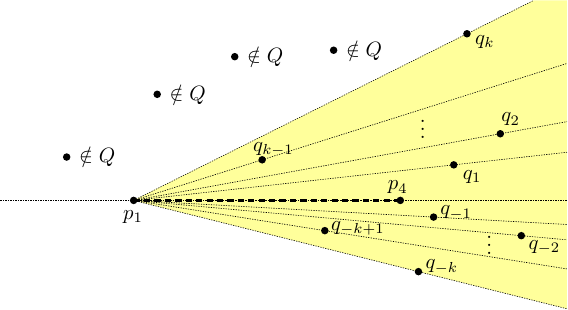}%
  \caption{Illustration for the proof of Lemma~\ref{lem:smallDrops}.}%
  \label{fig:smallDrops}%
\end{figure}

We sweep the points in $\myP \setminus \{\myp_4\}$ by angle from the ray $\myp_1\myp_4$. As shown in Figure~\ref{fig:smallDrops}, let $\myq_1,\ldots,\myq_\myk$ be the first $\myk$ points produced by this sweep in one direction, $\myq_{-1}\ldots,\myq_{-\myk}$ in the other direction and $\myQ=\{\myq_{-\myk}\ldots,\myq_{-1},\myq_1,\ldots,\myq_\myk\}$. To conclude the proof, we show that $\myp_3$ must be in $\myQ$. 
Suppose $\myp_3 \notin \myQ$ for the sake of a contradiction and assume without loss of generality that $\myp_3$ is on the side of $\myq_\myi$ with positive $\myi$. Then, consider the lines $\myL' = \{\myp_1\myq_1,\ldots,\myp_1\myq_\myk\}$.
Notice that $\myL' \subseteq \myL$, $\card{\myL'} = \myk$, and for each $\myl \in \myL'$ we have $\PotLine_\myl(\myS) > \PotLine_\myl(\myS')$, which contradicts the hypothesis that $\PotLine_\myL(\myS) - \PotLine_\myL(\myS') < \myk$. 
\end{proof}

Theorem~\ref{thm:distinct} is a consequence of Lemmas~\ref{lem:bigDrops} and~\ref{lem:smallDrops} with $\myk = \myn^{{1}/{3}}$.

\section{Untangling with Removal Choice}
\label{cha:R}
In this section, we devise strategies for removal choice to untangle multisets of segments with endpoints $\myP = \myC \cup \myT$ (where $\myC$ is in convex position), therefore providing upper bounds for several versions of $\dR$. Recall that such removal strategies choose which pair of crossing segments is removed, \emph{but not} which pair of segments with the same endpoints is subsequently inserted. We start with a point set in convex position, followed by $1$ point inside or outside the convex, then $1$ point inside and $1$ outside the convex, $2$ points inside the convex, and $2$ points outside the convex. As only removal choice is used, all results also apply to all versions.

\subsection{Upper Bounds for Convex Position}
\label{sec:RConvex}

Let $\myP = \myC = \{\myp_1,\ldots,\myp_{\card{\myC}}\}$ be a set of points in convex position sorted in counterclockwise order along the convex hull boundary and consider a set of segments $\myS$ with endpoints $\myP$. Given a segment $\sgt{\myp_\mya}{\myp_\myb}$ and assuming without loss of generality that $\mya<\myb$, we define the \emph{crossing depth} $\crossingDepth(\sgt{\myp_\mya}{\myp_\myb})$ as the number of points in $\myp_{\mya+1},\ldots,\myp_{\myb-1}$ that are an endpoint of a segment in $\myS$ that crosses any other segment in $\myS$ (not necessarily $\myp_\mya\myp_\myb$). We use the crossing depth to prove an $\OO(\myn \log \myn)$ bound in the convex $\Multigraph$ version.

\begin{figure}[htb]
    \hspace*{\stretch{1}}%
    \pbox[b]{\textwidth}{\centering\includegraphics[scale=\graphicsScale,page=1]{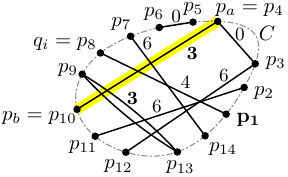}\newline(a)}\hspace*{\stretch{2}}%
    \pbox[b]{\textwidth}{\centering\includegraphics[scale=\graphicsScale,page=2]{convexR}\newline(b)}\hspace*{\stretch{2}}%
    \pbox[b]{\textwidth}{\centering\includegraphics[scale=\graphicsScale,page=3]{convexR}\newline(c)}\hspace*{\stretch{1}}%
    \caption{Proof of Theorem~\ref{thm:RUpperConvex}. (a) The segments of a convex multigraph are labeled with the crossing depth. (b,c) Two possible pairs of inserted segments, with one segment of the pair having crossing depth $\floor{\frac{3}{2}}=1$.}
    \label{fig:convexR}
\end{figure}

\begin{theorem}\label{thm:RUpperConvex}
Consider a multiset $\myS$ of $\myn$ segments with endpoints $\myP = \myC$ in satisfying the property $\Convex$.
There exists a removal strategy $\myR$ such that any untangle sequence of $\myR$ has length at most
\[\dR_{\Convex}(\myn) = \OO(\myn \log \card{\myC}) = \OO(\myn \log \myn).\]
\end{theorem}
\begin{proof}
We repeat the following procedure until there are no more crossings.
Let $\myp_\mya\myp_\myb \in \myS$ be a segment \emph{with crossings} (hence, crossing depth at least one) and $\mya<\myb$ minimizing $\crossingDepth(\myp_\mya\myp_\myb)$ (Figure~\ref{fig:convexR}(a)). Let $\myq_1,\ldots,\myq_{\crossingDepth(\myp_\mya\myp_\myb)}$ be the points defining $\crossingDepth(\myp_\mya\myp_\myb)$ in order and let $\myi = \ceil{\crossingDepth(\myp_\mya\myp_\myb)/2}$. Since $\myp_\mya\myp_\myb$ has minimum crossing depth, the point $\myq_\myi$ is the endpoint of segment $\myq_\myi\myp_\myc$ that crosses $\myp_\mya\myp_\myb$. When flipping $\myq_\myi\myp_\myc$ and $\myp_\mya\myp_\myb$, we obtain a segment $\mys$ (either $\mys=\myq_\myi\myp_\mya$ or $\mys=\myq_\myi\myp_\myb$) with $\crossingDepth(\mys)$ at most half of the original value of $\crossingDepth(\myp_\mya\myp_\myb)$ (Figure~\ref{fig:convexR}(b,c)). Hence, this operation always divides the value of the smallest positive crossing depth by at least two. As the crossing depth is an integer smaller than $\card{\myC}$, after performing this operation $\OO(\log \card{\myC})$ times, it produces a segment of crossing depth $0$. As the segments of crossing depth $0$ can no longer participate in a flip, the claimed bound follows.
\end{proof}

\subsection{Upper Bound for One Point Inside or Outside a Convex}
\label{sec:1InsideOutsideR}
In this section, we prove an upper bound on the number of flips to untangle a multiset of segments with all but one endpoint in convex position.
We first state a lemma used to prove Theorem~\ref{thm:1InsideOutsideR}.

\begin{lemma}
    \label{lem:farthestFirst}
    Consider a set $\myC$ of points in convex position, and a multiset $\myS$ of $\myn$ crossing-free segments with endpoints in $\myC$. Consider the multiset $\myS \cup \{\mys\}$ where $\mys$ is an extra segment with one endpoint in $\myC$ and one endpoint $\myq$ anywhere in the plane.
    There exists a removal strategy $R$ such that any untangle sequence in $R$ starting starting from $\myS \cup \{\mys\}$ has length at most $\myn - 1 = \OO(\myn)$.
\end{lemma}

\begin{figure}[htb]
    \hspace*{\stretch{1}}%
    \pbox[b]{\textwidth}{\centering\includegraphics[scale=\graphicsScale,page=1]{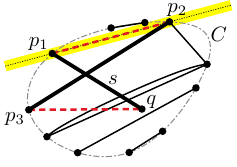}\newline(a)}\hspace*{\stretch{2}}%
    \pbox[b]{\textwidth}{\centering\includegraphics[scale=\graphicsScale,page=2]{1InsideOutsideConvexR}\newline(b)}\hspace*{\stretch{1}}%
    \caption{Proof of Lemma~\ref{lem:farthestFirst} with $\myq$ (a) inside, and (b) outside the convex hull of $\myC$.}
    \label{fig:farthestFirst}
\end{figure}

\begin{proof}
    Iteratively flip the segment $\mys=\myq\myp_1$ with the segment $\myp_2\myp_3 \in \myS$ crossing $\myq\myp_1$ the farthest from $\myq$ (Figure~\ref{fig:farthestFirst}).
    This flip inserts a $\myC\myC$-segment, say $\myp_1\myp_2$, which is impossible to flip again, because the line $\myp_1\myp_2$ is crossing free. The flip does not create any crossing between $\myC\myC$-segments of the multiset $\myS$. The last flip inserts two crossing-free segments instead of one, hence the $-1$ in the lemma statement.
\end{proof}

We are now ready to state and prove the theorem.

\begin{theorem}\label{thm:1InsideOutsideR}
Consider a multiset $\myS$ of $\myn$ segments such that the endpoints $\myP$ satisfy the property $\Property$ of being partitioned into $\myP = \myC \cup \myT$ where $\myC$ is in convex position, and $\myT = \{\myq\}$.
Let $\myt$ be the sum of the degrees of the points in $\myT$.
Let $\dR_{\Convex,\Gamma}(\myn)$ be the number of flips to untangle any multiset of at most $\myn$ segments with endpoints in convex position, a graph property $\Gamma$, and removal choice.
There exists a removal strategy $ R $ such that any untangle sequence in $ R $ has length at most
\[\dR_{\Property,\Gamma}(\myn) \leq \myt \cdot (\myn-1) + \dR_{\Convex,\Gamma}(\myn) = \OO(\myt\myn + \dR_{\Convex,\Gamma}(\myn)).\]

In particular, we have the following upper bounds for different graph properties:
\[\dR_{\Property}(\myn,\myt) = \OO(\myn \log \card{\myC} + \myt\myn) = \OO(\myn \log \myn + \myt\myn) \text{ and}\]
\[\dR_{\Property, \Cycle}(\myn,\myt) , \dR_{\Property, \RedBlue}(\myn,\myt) = \OO(\myt\myn).\]

Furthermore, if $\myS$ has no pair of $\myC\myC$-segments that cross each other, then any untangle sequence in $ R $ starting from $\myS$ has length at most $\myt\cdot(\myn-1)=\OO(\myt\myn)$.
\end{theorem}

\begin{figure}[htb]
    \hspace*{\stretch{1}}%
    \pbox[b]{\textwidth}{\centering\includegraphics[scale=\graphicsScale,page=3]{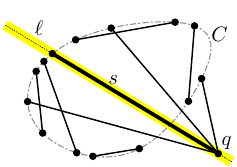}\newline(a)}\hspace*{\stretch{2}}%
    \pbox[b]{\textwidth}{\centering\includegraphics[scale=\graphicsScale,page=4]{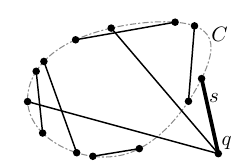}\newline(b)}\hspace*{\stretch{2}}%
    \pbox[b]{\textwidth}{\centering\includegraphics[scale=\graphicsScale,page=5]{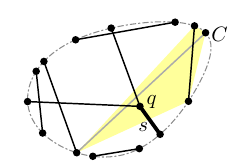}\newline(c)}\hspace*{\stretch{1}}%
    \caption{Proof of Theorem~\ref{thm:1InsideOutsideR}. (a) Case~1. (b) Case~2. (c) Case~3. The existence of the the gray segment is assumed for a contradiction.}
    \label{fig:1InsideOutsideR}
\end{figure}

\begin{proof}
We start by untangling the segment with both endpoints in $\myC$ using $\dR_{\Convex,\Gamma}(\myn)$ flips.
For each segment $\mys$ with endpoint $\myq$ with crossing, we apply Lemma~\ref{lem:farthestFirst} to $\mys$ and the $\myC\myC$-segments in $\myS$ crossing $\mys$.
Once a segment $\mys$ incident to $\myq$ is crossing free, it is impossible to flip it again as we fall in one of the following cases (Figure~\ref{fig:1InsideOutsideR}).
Let $\myl$ be the line containing $\mys$.

\textbf{Case~1:} If $\myl$ is crossing free, then $\myl$ splits (see Lemma~\ref{lem:splitting}) the multigraph in three partitions: the segments on one side of $\myl$, the segments on the other side of $\myl$, and the segment $\mys$ itself.

\textbf{Case~2:} If $\myl$ is not crossing free and $\myq$ is outside the convex hull of $\myC$, then $\mys$ is uncrossable (see the splitting lemma, i.e., Lemma~\ref{lem:splitting}).

\textbf{Case~3:} If $\myq$ is inside the convex hull of $\myC$, then introducing a crossing on $\mys$ would require that $\myq$ lies in the interior of the convex quadrilateral (shaded in yellow in Figure~\ref{fig:1InsideOutsideR}(c)) whose diagonals are the two segments removed by a flip. The procedure excludes this possibility by ensuring that there are no crossing pair of $\myC\myC$-segments in $\myS$, and, therefore, that one of the removed segment already has $\myq$ as an endpoint.

Therefore, we need at most $\myn - 1$ flips for each of the $\myt$ segments incident to $\myq$.
\end{proof}

\subsection[Upper Bound for One Point In and One Out a Convex]{Upper Bound for One Point Inside and One Point Outside a Convex}
\label{sec:1Inside1OutsideR}

Given an endpoint $\myp$, let $\degree(\myp)$ denote the degree of $\myp$, that is, the number of segments incident to $\myp$. The following lemma is used to prove Theorem~\ref{thm:1Inside1OutsideR}.

\begin{lemma} \label{lem:1Inside1OutsideR}
Consider a multiset $\myS$ of $\myn$ segments with endpoints $\myP$ partitioned into $\myP = \myC \cup \myT$ where $\myC$ is in convex position, and $\myT = \{\myq,\myq'\}$ such that $\myq$ is outside the convex hull of $\myC$ and $\myq'$ is inside the convex hull of $\myC$.
Consider that $\myq$ is the endpoint of a single segment $\mys$ in $\myS$ and that all the crossings of $\myS$ are on $\mys$.
We define the parameter $\myt$ as the sum of the degrees of the points in $\myT$, i.e., $ \myt = \degree ( \myq ' ) + 1 $.

There exists a removal strategy $\myR$ such that any flip sequence of $\myR$ starting at $\myS$ and ending in a multiset of segments where all crossings (if any) are on the segment $\myq\myq'$ (if $\sgt{\myq}{\myq'} \notin \myS$ then there are no crossings) has length $\OO(\myt\myn)$.
\end{lemma}
\begin{proof}
We proceed as follows, while $\mys$ has crossings. For induction purpose, let $\nbf(\myn)$ be the length of the flip sequence in the lemma statement for $\myn$ segments.
Let $\mys'$ be the segment that crosses $\mys$ at the point farthest from $\myq$. We flip $\mys$ and $\mys'$, arriving at one of the three cases below (Figure~\ref{fig:1Inside1OutsideR}).

\begin{figure}[!ht]
    \hspace*{\stretch{1}}%
    \pbox[b]{\textwidth}{\centering\includegraphics[scale=\graphicsScale,page=1]{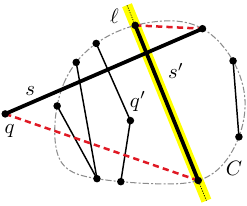}\newline Case~1}\hspace*{\stretch{2}}%
    \pbox[b]{\textwidth}{\centering\includegraphics[scale=\graphicsScale,page=2]{1Inside1OutsideR}\newline Case~2}\hspace*{\stretch{2}}%
    \pbox[b]{\textwidth}{\centering\includegraphics[scale=\graphicsScale,page=3]{1Inside1OutsideR}\newline Case~3}\hspace*{\stretch{1}}%
    \caption{The three cases in the proof of Lemma~\ref{lem:1Inside1OutsideR}.}
    \label{fig:1Inside1OutsideR}%
\end{figure}

\paragraph{Case~1 ($\mathbf{\myC\myT \times \myC\myC)}$.} In this case, the segment $\mys'$ is a $\myC\myC$-segment. Notice that the line $\myl$ containing $\mys'$ becomes crossing free after the flip. There are segments on both sides of $\myl$. 
If $\myl$ separates $\myq,\myq'$, then we untangle both sides independently (see  Lemma~\ref{lem:splitting}) using $\OO(\myn)$ and $\OO(\myt \myn)$ flips (Theorem~\ref{thm:1InsideOutsideR}). Otherwise, the segments on one side of $\myl$ are already crossing free (because of the specific choice of $\mys'$) and we inductively untangle the $\myn' \leq \myn-1$ segments on the other side of $\myl$ using $\nbf(\myn')$ flips.

\paragraph{Case~2 ($\mathbf{\myC\myT \times \myC\myT \rightarrow \myC\myC,\myT\myT}$).} If $\mys'$ is a $\myC\myT$-segment and one of the inserted segments is the $\myT\myT$-segment $\myq\myq'$, then the procedure is over as all crossings are on $\myq\myq'$.

\paragraph{Case~3 ($\mathbf{\myC\myT \times \myC\myT \rightarrow \myC\myT,\myC\myT}$).} In this case two $\myC\myT$-segments are inserted. Let $\myp \in \myC$ be an endpoint of $\mys = \myq\myp$. Since the inserted $\myC\myT$-segment $\myq'\myp$ is crossing free, Case~3 only happens $\OO(\myt)$ times before we arrive at Case~1 or Case~2.

Putting the three cases together, we obtain the recurrence
\[\nbf(\myn) \leq \OO(\myt) + \nbf(\myn')\text{, with } \myn' \leq \myn-1,\]
which solves to $\nbf(\myn) = \OO(\myt\myn)$, as claimed.
\end{proof}

We are now ready to prove the theorem.

\begin{theorem}\label{thm:1Inside1OutsideR}
Consider a multiset $\myS$ of $\myn$ segments with endpoints $\myP$, such that $\myP$ satisfies the property $\Property$ that $\myP$ is partitioned into $\myP = \myC \cup \myT$ where $\myC$ is in convex position, and $\myT = \{\myq,\myq'\}$ with $\myq$ outside the convex hull of $\myC$ and $\myq'$ inside the convex hull of $\myC$.
Let $\myt$ be the sum of the degrees of the points in $\myT$, i.e., $\myt = \degree(\myq) + \degree(\myq')$.
Let $\dR_{\Convex,\Gamma}(\myn)$ be the number of flips to untangle any multiset of at most $\myn$ segments with endpoints in convex position, a graph property $\Gamma$, and removal choice.
There exists a removal strategy $\myR$ such that any untangle sequence of $\myR$ for the graph property $\Gamma$ has length
\[\dR_{\Property,\Gamma}(\myn,\myt) = \OO(\degree(\myq)\degree(\myq')\myn + \dR_{\Convex,\Gamma}(\myn)) = \OO(\myt^2\myn + \dR_{\Convex,\Gamma}(\myn)).\]

In particular, we have the following upper bounds for different graph properties:
\[\dR_{\Property}(\myn,\myt) = \OO(\degree(\myq)\degree(\myq')\myn + \myn \log \myn) = \OO(\myt^2\myn + \myn \log \myn) \text{ and }\]
\[\dR_{\Property,\Cycle}(\myn) , \dR_{\Property,\RedBlue}(\myn) = \OO(\myn) .\]
\end{theorem}
\begin{proof}
We prove the theorem in the most general case where $ \Gamma = \Multigraph$.
The untangle sequence is decomposed into four phases.

\textbf{Phase~1 ($\mathbf{\myC\myC\times \myC\myC}$).} In this phase, we remove all crossings between pairs of $\myC\myC$-segments in $\myS$ using $\dR_{\Convex,\Gamma}(\myn)$ flips. Throughout all the phases, the invariant that no pair of $\myC\myC$-segments in $\myS$ crosses is preserved.

\textbf{Phase~2 ($\mathbf{\myC\myq' \times \myC\myC}$).} In this phase, we remove all crossings between pairs composed of a $\myC\myC$-segment and a $\myC\myT$-segment incident to $\myq'$ (the point inside the convex hull of $\myC$) using $\OO(\myt\myn)$ flips by Theorem~\ref{thm:1InsideOutsideR}.

\textbf{Phase~3 ($\mathbf{\myC\myq}$).}
At this point, all crossings involve a segment incident to $\myq$. In this phase, we deal with all remaining crossings except the crossings involving the segment $\myq\myq'$. Lemma~\ref{lem:splitting} allows us to remove the crossings in each $\myC\myT$-segment $\mys$ incident to $\myq$ independently, which we do using $\OO(\degree(\myq') \myn)$ flips using Lemma~\ref{lem:1Inside1OutsideR}.
As there are $\degree(\myq)$ $\myC\myT$-segments adjacent to $\myq$ in $\myS$, the total number of flips is $\OO(\degree(\myq) \degree(\myq') \myn) = \OO(\myt^2\myn)$.

\textbf{Phase~4 ($\mathbf{\myC\myC \times \myT\myT}$).} At this point, all crossings involve the $\myT\myT$-segment $\myq\myq'$. The endpoints in $\myC$ that are adjacent to segments with crossings, together with $\myq'$, are all in convex position. Hence, the only endpoint not in convex position is $\myq$, and we apply Theorem~\ref{thm:1InsideOutsideR} using $\OO(\myt\myn)$ flips. 

After the $\dR_{\Convex,\Gamma}(\myn)$ flips in Phase~1, the number of flips is dominated by Phase~3 with $\OO(\degree(\myq) \degree(\myq') \myn) = \OO(\myt^2\myn)$ flips.
\end{proof}

Notice that, in certain cases (for example in the red-blue case with $\myq,\myq'$ having different colors) a flip between two $\myC\myT$-segments never produces two $\myC\myT$-segments. Consequently, Case~3 of the proof of Lemma~\ref{lem:1Inside1OutsideR} never happens, and the bound in Theorem~\ref{thm:1Inside1OutsideR} decreases to $\OO(\dR_{\Convex,\Gamma}(\myn) + \myt\myn)$.

\subsection{Upper Bound for Two Points Inside a Convex}
\label{sec:2insideR}

We prove a similar theorem for two points inside the convex hull of $\myC$.

\begin{theorem}\label{thm:2InsideR}
Consider a multiset $\myS$ of $\myn$ segments with endpoints $\myP$, such that $\myP$ satisfies the property $\Property$ that $\myP$ is partitioned into $\myP = \myC \cup \myT$ where $\myC$ is in convex position, and $\myT = \{\myq,\myq'\}$ with $\myq,\myq'$ inside the convex hull of $\myC$.
Let $\myt$ be the sum of the degrees of the points in $\myT$, i.e., $\myt = \degree(\myq) + \degree(\myq')$.
Let $\dR_{\Convex,\Gamma}(\myn)$ be the number of flips to untangle any multiset of at most $\myn$ segments with endpoints in convex position, a graph property $\Gamma$, and removal choice.
There exists a removal strategy $\myR$ such that any untangle sequence of $\myR$ for the graph property $\Gamma$ has length
\[\dR_{\Property,\Gamma}(\myn,\myt) = \OO(\myt\myn + \dR_{\Convex,\Gamma}(\myn)).\]

In particular, we have the following upper bounds for different graph properties:
\[\dR_{\Property}(\myn,\myt) = \OO(\myt\myn + \myn \log \myn) \text{ and }\]
\[\dR_{\Property,\Cycle}(\myn) , \dR_{\Property,\RedBlue}(\myn) = \OO(\myn) .\]
\end{theorem}
\begin{proof}
The untangle sequence is decomposed in five phases. At the end of each phase, a new type of crossings is removed, and types of crossings removed in the previous phases are not present, even if they may temporarily appear during the phase.

\textbf{Phase~1 ($\mathbf{\myC\myT\times \myC\myT}$).} In this phase, we remove all crossings between pairs of $\myC\myT$-segments in $\myS$ using $\OO(\dR_{\Convex,\Gamma}(\myt)) = \OO(\dR_{\Convex,\Gamma}(\myn))$ flips. We separately solve two convex sub-problems defined by the $\myC\myT$-segments in $\myS$, one on each side of the line $\myq\myq'$.

\textbf{Phase~2 ($\mathbf{\myC\myC\times \myC\myC}$).} In this phase, we remove all crossings between pairs of $\myC\myC$-segments in $\myS$ using $\OO(\dR_{\Convex,\Gamma}(\myn))$ flips. As no $\myC\myT$-segment has been created, there is still no crossing between a pair of $\myC\myT$ segments. Throughout, our removal will preserve the invariant that no pair of $\myC\myC$-segments in $\myS$ crosses.

\textbf{Phase~3 ($\mathbf{\myC\myT \times \textbf{non-central } \myC\myC}$).} 
We distinguish between a few types of $\myC\myC$-segments. The \textit{central} $\myC\myC$-segments cross the segment $\myq\myq'$ (regardless of $\myq\myq'$ being in $\myS$ or not), while the \emph{non-central} do not. The \textit{peripheral} $\myC\myC$-segments cross the line $\myq\myq'$ but not the segment $\myq\myq'$, while the \emph{outermost} $\myC\myC$-segments do not cross either. In this phase, we remove all crossings between $\myC\myT$-segments in $\myS$ and non-central $\myC\myC$-segments in $\myS$.

Given a non-central $\myC\myC$-segment $\myp\myp'$, let the \emph{out-depth} of the segment $\sgt{\myp}{\myp'}$ be the number of points of $\myC$ that are contained inside the halfplane bounded by the line $\myp\myp'$ and not containing $\myT$. Also, let $\PotCNC$ be the number of crossings between the non-central $\myC\myC$-segments and the $\myC\myT$-segments in $\myS$. At the end of each step the two following invariants are preserved. (i) No pair of $\myC\myC$-segments in $\myS$ crosses. (ii) No pair of $\myC\myT$-segments in $\myS$ crosses.

At each step, we choose to flip a non-central $\myC\myC$-segment $\myp\myp'$ of minimum out-depth that crosses a $\myC\myT$-segment. We flip $\myp\myp'$ with the $\myC\myT$-segment $\myq''\myp''$ (with $\myq'' \in \{\myq,\myq'\}$)
that crosses $\myp\myp'$ at the point closest to $\myp$ (Figure~\ref{fig:2InsideR-3-1}(a) and Figure~\ref{fig:2InsideR-3-2}(a)).
One of the possibly inserted pairs may contain a $\myC\myT$-segment $\mys$ that crosses another $\myC\myT$-segment $\mys'$, violating the invariant (ii) (Figure~\ref{fig:2InsideR-3-1}(b) and Figure~\ref{fig:2InsideR-3-2}(b)). If there are multiple such segments $\mys'$, then we consider $\mys'$ to be the segment whose crossing with $\mys$ is closer to $\myq''$. We flip $\mys$ and $\mys'$ and obtain either two $\myC\myT$-segments (Figure~\ref{fig:2InsideR-3-1}(c) and Figure~\ref{fig:2InsideR-3-2}(c)) or a $\myC\myC$-segment and the segment $\myq\myq'$ (Figure~\ref{fig:2InsideR-3-1}(d) and Figure~\ref{fig:2InsideR-3-2}(d)). The analysis is divided in two main cases.

\begin{figure}[htb]
    \hspace*{\stretch{1}}%
    \pbox[b]{\textwidth}{\centering\includegraphics[scale=\graphicsScale,page=2]{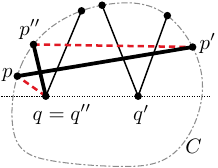}\newline(a)}\hspace*{\stretch{2}}%
    \pbox[b]{\textwidth}{\centering\includegraphics[scale=\graphicsScale,page=3]{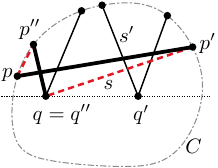}\newline(b)}\hspace*{\stretch{2}}%
    \pbox[b]{\textwidth}{\centering\includegraphics[scale=\graphicsScale,page=5]{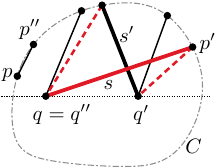}\newline(c)}\hspace*{\stretch{2}}%
    \pbox[b]{\textwidth}{\centering\includegraphics[scale=\graphicsScale,page=4]{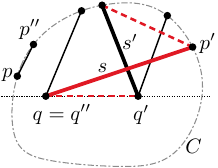}\newline(d)}\hspace*{\stretch{1}}%
    \caption{Theorem~\ref{thm:2InsideR}, Phase~3 when $\myp\myp'$ is an outermost segment.}
    \label{fig:2InsideR-3-1}
\end{figure}

If $\myp\myp'$ is an outermost $\myC\myC$-segment (see Figure~\ref{fig:2InsideR-3-1}), then case analysis shows that the two invariants are preserved and $\PotCNC$ decreases. 

\begin{figure}[htb]
    \hspace*{\stretch{1}}%
    \pbox[b]{\textwidth}{\centering\includegraphics[scale=\graphicsScale,page=2]{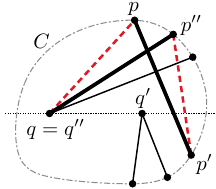}\newline(a)}\hspace*{\stretch{2}}%
    \pbox[b]{\textwidth}{\centering\includegraphics[scale=\graphicsScale,page=3]{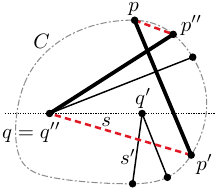}\newline(b)}\hspace*{\stretch{2}}%
    \pbox[b]{\textwidth}{\centering\includegraphics[scale=\graphicsScale,page=5]{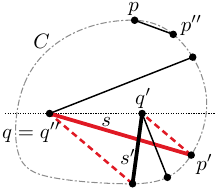}\newline(c)}\hspace*{\stretch{2}}%
    \pbox[b]{\textwidth}{\centering\includegraphics[scale=\graphicsScale,page=4]{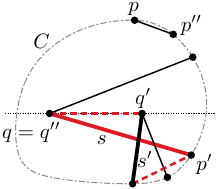}\newline(d)}\hspace*{\stretch{1}}%
    \caption{Theorem~\ref{thm:2InsideR}, Phase~3 when $\myp\myp'$ is a peripheral segment.}
    \label{fig:2InsideR-3-2}
\end{figure}

If $\myp\myp'$ is a peripheral $\myC\myC$-segment (see Figure~\ref{fig:2InsideR-3-2}), then a case analysis shows that the two invariants are preserved and $\PotCNC$ has the following behavior. If no $\myC\myC$-segment is inserted, then $\PotCNC$ decreases. Otherwise a $\myC\myC$-segment and a $\myT\myT$-segment are inserted and $\PotCNC$ may increase by $\OO(\myt)$ (Figure~\ref{fig:2InsideR-3-2}(d)). Notice that the number of times the $\myT\myT$-segment $\myq\myq'$ is inserted is $\OO(\myt)$, which bounds the total increase by $\OO(\myt^2)$. 

As $\PotCNC = \OO(\myt\myn)$, the total increase is $\OO(\myt^2)$, and $\PotCNC$ decreases at all but $\OO(\myt)$ steps, we have that the number of flips in Phase~3 is $\OO(\myt\myn)$.

\textbf{Phase~4 ($\mathbf{\myC\myT \times \myC\myC\central}$).} 
At this point, each crossing involves a central $\myC\myC$-segment and either a $\myC\myT$-segment or the $\myT\myT$-segment $\myq\myq'$.
In this phase, we remove all crossings between $\myC\myT$-segments and central $\myC\myC$-segments in $\myS$, ignoring the $\myT\myT$-segments.
This phase ends with crossings only between $\myq\myq'$ and central $\myC\myC$-segments in $\myS$.

\begin{figure}[htb]
    \hspace*{\stretch{1}}%
    \pbox[b]{\textwidth}{\centering\includegraphics[scale=\graphicsScale,page=1]{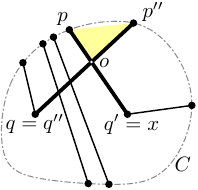}\newline(a)}\hspace*{\stretch{2}}%
    \pbox[b]{\textwidth}{\centering\includegraphics[scale=\graphicsScale,page=2]{2InsideRphase4}\newline(b)}\hspace*{\stretch{2}}%
    \pbox[b]{\textwidth}{\centering\includegraphics[scale=\graphicsScale,page=3]{2InsideRphase4}\newline(c)}\hspace*{\stretch{2}}%
    \pbox[b]{\textwidth}{\centering\includegraphics[scale=\graphicsScale,page=4]{2InsideRphase4}\newline(d)}\hspace*{\stretch{1}}%
    \caption{Theorem~\ref{thm:2InsideR}, Phase~4. (a) A pair of $\myC\myT$-segments with an ear. (b) A $\myC\myC$-segment and a $\myC\myT$-segment with an ear. (c) Flipping an ear that produces crossing pairs of $\myC\myT$-segments. (d) Flipping an ear that inserts a non-central $\myC\myC$-segment with crossings.}
    \label{fig:2InsideR-4}
\end{figure}

Given four endpoints $\myq'' \in \myT$, $\myp,\myp'' \in \myC$, and $\myx \in \myC \cup \myT$, we say that a pair of segments $\myp''\myq'',\myx\myp \in \myS$ crossing at a point $\myo$ contains an \emph{ear} $\widehat{\myp\myp''}$ if the interior of the triangle $\myp\myp''\myo$ intersects no segment of $\myS$ (see Figure~\ref{fig:2InsideR-4}(a) and~\ref{fig:2InsideR-4}(b)).
Every set of segments with endpoints in $\myC \cup \myT$ with $\card{\myT} = 2$ that has crossings (not involving the $\myT\myT$-segment) contains an ear (adjacent to the crossing that is farthest from the line $\myq\myq'$).

At each \emph{step}, we flip a pair of segments $\myp''\myq'',\myx\myp$ that contains an ear $\widehat{\myp\myp''}$, prioritizing pairs where both segments are $\myC\myT$-segments. Notice that, even though initially we did not have crossing pairs of $\myC\myT$-segments in $\myS$, they may be produced in the flip (Figure~\ref{fig:2InsideR-4}(c)).
If the flip inserts a non-central $\myC\myC$-segment which crosses some $\myC\myT$-segments in $\myS$ (Figure~\ref{fig:2InsideR-4}(d)), then, we perform the following \emph{while loop}. Assume without loss of generality that $\myq\myq'$ is horizontal and $\mys$ is closer to $\myq'$ than to $\myq$. While there exists a non-central $\myC\myC$-segment $\mys$ with crossings, we flip $\mys$ with the $\myC\myT$-segment $\mys'$ crossing $\mys$ that comes first according to the following order. As a first criterion, a segment incident to $\myq$ comes before a segment incident to $\myq'$. As a second tie-breaking criterion, a segment whose crossing point with $\mys$ that is farther from the line $\myq\myq'$ comes before one that is closer.

Let $\PotCTCNC = \OO(\myt\myn)$ be the number of crossings between central $\myC\myC$-segments and $\myC\myT$-segments in $\myS$ plus the number of crossings between $\myC\myT$-segments in $\myS$.
A case analysis shows that the value of $\PotCTCNC$ decreases at each step. If no non-central $\myC\myC$-segment is inserted, then the corresponding step consists of a single flip. As $\PotCTCNC$ decreases, there are $\OO(\myt\myn)$ steps that do not insert a non-central $\myC\myC$-segment.

However, if a non-central $\myC\myC$-segment is inserted, at the end of the step we inserted a $\myC\myC$-segment that is uncrossable (see Lemma~\ref{lem:splitting}). As the number of $\myC\myC$-segments in $\myS$ is $\OO(\myn)$, we have that the number of times the while loop is executed is $\OO(\myn)$. Since each execution of the while loop performs $\OO(\myt)$ flips, we have a total of $\OO(\myt\myn)$ flips in this phase.

\textbf{Phase~5 ($\mathbf{\myT\myT \times \myC\myC\central}$).} In this phase, we remove all crossings left, which are between the possibly multiple copies of the $\myT\myT$-segment $\myq\myq'$ and central $\myC\myC$-segments in $\myS$. The endpoints of the segments with crossings are in convex position and all other endpoints are outside their convex hull. Hence, by Lemma~\ref{lem:splitting}, it is possible to obtain a crossing-free multigraph using $\OO(\dR_{\Convex,\Gamma}(\myn))$ flips.
\end{proof}

\subsection{Upper Bound for Two Points Outside a Convex}
\label{sec:2OutsideR}

In this section, we consider the case of two endpoints $\myq,\myq'$ outside the convex hull of the remaining endpoints $\myC$, which are in convex position. We prove a theorem with a bound that is quadratic in $\myt = \degree(\myq) + \degree(\myq')$ but linear or near-linear in $\myn$. 

\begin{theorem}\label{thm:2OutsideR}
Consider a multiset $\myS$ of $\myn$ segments with endpoints $\myP$, such that $\myP$ satisfies the property $\Property$ that $\myP$ is partitioned into $\myP = \myC \cup \myT$ where $\myC$ is in convex position, and $\myT = \{\myq,\myq'\}$ with $\myq,\myq'$ outside the convex hull of $\myC$.
Let $\myt$ be the sum of the degrees of the points in $\myT$, i.e., $\myt = \degree(\myq) + \degree(\myq')$.
Let $\dR_{\Convex,\Gamma}(\myn)$ be the number of flips to untangle any multiset of at most $\myn$ segments with endpoints in convex position, a graph property $\Gamma$, and removal choice.
There exists a removal strategy $\myR$ such that any untangle sequence of $\myR$ for the graph property $\Gamma$ has length at most
\[\dR_{\Property,\Gamma}(\myn,\myt) = \OO(\degree(\myq) \degree(\myq') \myn ) + \dR_{\Convex,\Gamma}(\myn) = \OO(\myt^2\myn + \dR_{\Convex,\Gamma}(\myn)).\]

In particular, we have the following upper bounds for different graph properties:
\[\dR_{\Property}(\myn,\myt) = \OO(\myt^2\myn + \myn \log \myn) \text{ and }\]
\[\dR_{\Property,\Cycle}(\myn) , \dR_{\Property,\RedBlue}(\myn) = \OO(\myn) .\]
\end{theorem}

\begin{figure}[htb]
    \centering
    \begin{tabularx}{\textwidth}{CCCC}
        \includegraphics[scale=\graphicsScale,page=1]{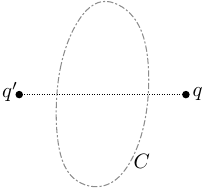} &%
        \includegraphics[scale=\graphicsScale,page=2]{2OutsideRMultigraphRevised} &%
        \includegraphics[scale=\graphicsScale,page=3]{2OutsideRMultigraphRevised} &%
        \includegraphics[scale=\graphicsScale,page=4]{2OutsideRMultigraphRevised} %
        \\
        (a) & (b) & (c) & (d)
    \end{tabularx}
 \caption{Illustration of the four case of relative position of the points $\myq,\myq'$ and the convex hull of $\myC$ for the proof of Theorem~\ref{thm:2OutsideR}.
 Without loss of generality, the case (d) is not considered.}%
 \label{fig:2OutsideR_WLOG}
\end{figure}

Without loss of generality, we assume that $\myq'$ is not in the interior of the convex hull of $\myC \cup \{\myq\}$ (Figure~\ref{fig:2OutsideR_WLOG}).
The strategy starts with the following preprocessing phase. We \emph{delete} from $\myS$ all segments incident to $\myq$. We then untangle $\myS$ using Theorem~\ref{thm:1InsideOutsideR}, which uses at most $(\myn-1) \degree(\myq') + \dR_{\Convex,\Gamma}(\myn) = \OO(\myn \degree(\myq')) + \dR_{\Convex,\Gamma}(\myn)$ flips.

Next, we \emph{add back} to $\myS$ all segments incident to $\myq$, one by one, in any order. When adding back a $\myT\myT$-segment (Figure~\ref{fig:2OutsideR_TT}(a)), we call Routine~$\myT\myT$, which will untangle $\myS$ using $\OO(\myn)$ flips (Figure~\ref{fig:2OutsideR_TT}(d)). When adding back a $\myC\myT$-segment (Figure~\ref{fig:2OutsideR_CT1}(a) and Figure~\ref{fig:2OutsideR_CT2}(a)), we call Routine~$\myC\myT$, which will untangle $\myS$ using $\OO(\myn \degree(\myq'))$ flips. Since the number of segments that we add back is $\degree(\myq)$, the bound stated in Theorem~\ref{thm:2OutsideR} holds. Next, we describe Routine~$\myT\myT$, followed by Routine~$\myC\myT$.

\begin{figure}[htb]
    \centering
    \begin{tabularx}{\textwidth}{CCCC}
        \includegraphics[scale=\graphicsScale,page=5]{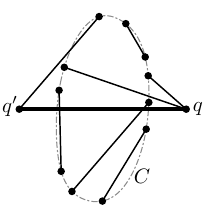} &%
        \includegraphics[scale=\graphicsScale,page=6]{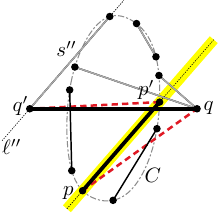} &%
        \includegraphics[scale=\graphicsScale,page=7]{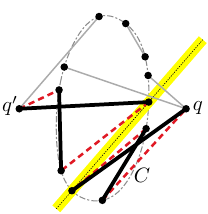} &%
        \includegraphics[scale=\graphicsScale,page=8]{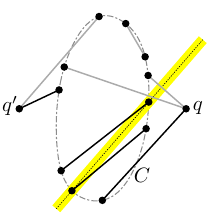}%
        \\
        (a) & (b) & (c) & (d)
    \end{tabularx}
 \caption{Illustration for Routine~$\myT\myT$ in the proof of Theorem~\ref{thm:2OutsideR}. (a) A $\myT\myT$-segment (in bold) is added back to $\myS$. (b) The first flip (the flipped segments are in bold, the inserted segments are dashed in red). The segments of $\myS$ that have been removed from $\myS'$ are grayed. The splitting line is highlighted. (c) The remaining flips. (d) $\myS$ is crossing free at the end of Routine~$\myT\myT$.}%
 \label{fig:2OutsideR_TT}
\end{figure}

\paragraph{Routine~$\myT\myT$.}
If $\sgt{\myq}{\myq'}$ has no crossing, then we are done as $\myS$ is crossing free. Otherwise, we know that there is a single copy of $\sgt{\myq}{\myq'}$, because $\myS \setminus \{\sgt{\myq}{\myq'}\}$ is crossing free.
We start by flipping $\sgt{\myq}{\myq'}$ with an arbitrary segment $\sgt{\myp}{\myp'}$ that crosses $\sgt{\myq}{\myq'}$. Assume without loss of generality that the flip inserted $\sgt{\myp}{\myq}$ and $\sgt{\myp'}{\myq'}$ (Figure~\ref{fig:2OutsideR_TT}(b)).

We now construct a splitting partition of $\myS$ in order to apply Lemma~\ref{lem:splitting}. Let $\myS'$ be the set of segments that have not been assigned a partition yet. Initially $\myS' = \myS$.  We say that a segment is \emph{$\myS'$-uncrossable} if it is not crossed by any segment defined by two endpoints of segments in $\myS'$. While there is an $\myS'$-uncrossable segment $\myu \in \myS'$, we create a new singleton partition containing $\myu$, removing $\myu$ from $\myS'$. We claim that at the end of this loop, the line $\lineT{\myp}{\myp'}$ is crossing free (Figure~\ref{fig:2OutsideR_TT}(b)). We then create two separate partitions $\myS_1,\myS_2$ with the segments of $\myS'$ on each side of $\lineT{\myp}{\myp'}$, as well as a partition with the copies of the segment $\sgt{\myp}{\myp'}$ that are still in $\myS'$ (if any). We are now ready to apply Lemma~\ref{lem:splitting}. The only partitions that may have crossings are $\myS_1,\myS_2$. Each of $\myS_1,\myS_2$ have all crossings in the single $\myC\myT$-segment. Hence, we apply Theorem~\ref{thm:1InsideOutsideR} with $\myt=1$ to untangle each of $\myS_1,\myS_2$ using $\OO(\myn)$ flips, and consequently untangling $\myS$ using $\OO(\myn)$ flips (Figure~\ref{fig:2OutsideR_TT}(c)).

We say that a $\myC\myT$-segment is a \emph{$\myC\myT\inner$-segment} if it intersects the interior of the convex hull of $\myC$ and a \emph{$\myC\myT\outter$-segment} otherwise.
To prove the claim, first note that, since the segment $\sgt{\myq}{\myq'}$ intersects the convex hull of $\myC$, all the $\myC\myT\outter$-segments are initially (i.e., when $\myS'=\myS$) uncrossable, and thereby removed first from $\myS'$.
Second, after all the $\myC\myT\outter$-segments have been removed from $\myS'$, consider the $\myC\myT\inner$-segment $\mys''$ of $\myS'$ whose endpoint in $\myC$ is the farthest away from the line $\lineT{\myq}{\myq'}$.
The line $\myl''$ containing $\mys''$ is crossing free (Figure~\ref{fig:2OutsideR_TT}(b)).
Thus, by Lemma~\ref{lem:splitting}, we partition $\myS'$ in three: the segments on each side of the line $\myl''$ and the segments on the line $\myl''$ itself.
The segments on the side of $\myl''$ which does not contain the open segment $\sgt{\myq}{\myq'}$ are crossing free, and are thus removed from $\myS'$.
The claim follows, since at the end of the loop, all the $\myC\myT$-segments have been removed from $\myS'$.

\begin{figure}[htb]
    \centering
    \begin{tabularx}{\textwidth}{CCCC}
        \includegraphics[scale=\graphicsScale,page=9]{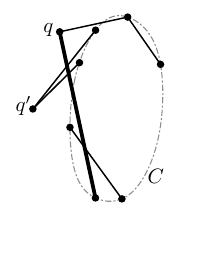} &%
        \includegraphics[scale=\graphicsScale,page=10]{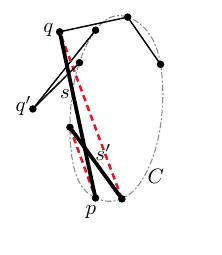} &%
        \includegraphics[scale=\graphicsScale,page=11]{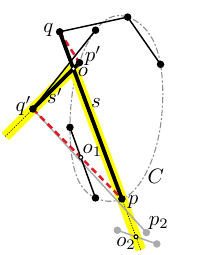} &%
        \includegraphics[scale=\graphicsScale,page=12]{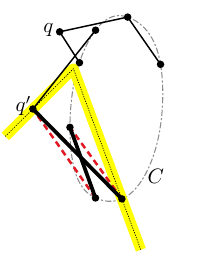}%
        \\
        (a) & (b) & (c) & (d)
    \end{tabularx}
 \caption{Illustrations for Routine~$\myC\myT$ in the proof of Theorem~\ref{thm:2OutsideR}. (a)~A $\myC\myT$-segment (in bold) is added back to $\myS$. (b)~The case where $\mys'$ is a $\myC\myC$-segment. (c)~The case where $\mys'$ is a $\myC\myT$-segment. The existence of the two grayed segments is assumed for a contradiction. (d)~Illustration of the subroutine Restore-Loop-Invariant.}%
 \label{fig:2OutsideR_CT1}
\end{figure}

\paragraph{Routine~$\myC\myT$.}
This routine untangles $\myS$ in the case where the segment $\mys$ that we add back to $\myS$ is a $\myC\myT$-segment $\mys=\sgt{\myp}{\myq}$ (Figure~\ref{fig:2OutsideR_CT1}(a) and Figure~\ref{fig:2OutsideR_CT2}(a)).
This routine consists of the following while loop and calls a subroutine called Restore-Loop-Invariant to maintain the invariant that removing a single $\myC\myT$-segment from $\myS$ would remove all the crossings.

While the segment $\mys$ crosses some segment in $\myS$:
\begin{enumerate}
\item Let $\mys'$ be the segment of $\myS$ that crosses $\mys$ \emph{the farthest away from $\myq$}.
\item\label{line:flip} Flip $\mys$ with $\mys'$ (see Figure~\ref{fig:2OutsideR_CT1}(b,c) and Figure~\ref{fig:2OutsideR_CT2}(b,c,d) for five different cases of this flip).
\item If $\mys'$ is not a $\myC\myC$-segment, then call Restore-Loop-Invariant.
\item Set $\mys$ to be the segment incident to $\myq$ that was inserted by the flip of $\mys$ with $\mys'$.
\end{enumerate}

\paragraph{Restore-Loop-Invariant.} 

Let $\myo$ be the crossing point of $\mys=\sgt{\myq}{\myp}$ and $\mys'=\sgt{\myp'}{\myq'}$.
\begin{enumerate}
    \item\label{case:qq'NotInserted} If the inserted pair is $\sgt{\myq}{\myp'},\sgt{\myp}{\myq'}$ (Figure~\ref{fig:2OutsideR_CT1}(c) and Figure~\ref{fig:2OutsideR_CT2}(b)):
    \begin{enumerate}
        \item\label{case:qq'NotInserted(a)} If $\sgt{\myp}{\myq'}$ has crossings, then untangle the convex sector defined by the rays $\ray{\myo}{\myq'}$ and $\ray{\myo}{\myp}$, i.e., the rays with origin $\myo$ and directed by the vectors $\vect{\myo}{\myq'}$ and $\vect{\myo}{\myp}$ (Figure~\ref{fig:2OutsideR_CT1}(d)). We will show that these rays are crossing free.
        We then apply Theorem~\ref{thm:1InsideOutsideR} in the case where the $\myC\myC$-segments are already crossing free and $\myt=1$ (as the other segments incident to $\myq'$ are already crossing free), using $\OO(\myn)$ flips.
    \end{enumerate}
    \item\label{case:qq'Inserted} If the inserted pair is $\sgt{\myq}{\myq'},\sgt{\myp}{\myp'}$ (we will show that the following cases are mutually exclusive):
    \begin{enumerate}
        \item\label{case:qq'HasCrossings} If $\sgt{\myq}{\myq'}$ has crossings (Figure~\ref{fig:2OutsideR_CT2}(c)), then call Routine~$\myT\myT$ (at most $\myn$ flips).
        \item\label{case:pp'HasCrossings} If $\sgt{\myp}{\myp'}$ has crossings (Figure~\ref{fig:2OutsideR_CT2}(d)), then $\sgt{\myp}{\myp'}$ must cross another $\myC\myC$-segment. While there is a crossing between two $\myC\myC$-segments, we flip them, in any order. This procedure does not create crossings involving $\myC\myT$-segments or $\myT\myT$-segments, hence the number of crossings strictly decreases. Consequently, the number of flips performed is at most the number of crossings of the segment $\sgt{\myp}{\myp'}$ when the routine was called, which is $\OO(\myn)$.
    \end{enumerate}
\end{enumerate}

Next, we show that the rays $\ray{\myo}{\myq'}$ and $\ray{\myo}{\myp}$ in step~\ref{case:qq'NotInserted(a)} are crossing free, which ensures that flips performed in this convex sector do not insert segments that cross the remaining segments. The proof breaks the two rays into two parts each.

Since the segment $\mys'=\sgt{\myp'}{\myq'}$ is crossing free, then the segment $\sgt{\myo}{\myq'}$ is crossing free.
Since $\mys'$ is the segment that crosses $\mys$ the farthest away from $\myq$, then the segment $\sgt{\myo}{\myp}$ is crossing free.

At this point, we know that all the segments of $\myS$ crossing the segment $\sgt{\myp}{\myq'}$ are $\myC\myC$-segments. Let $\myo_1$ be the intersection point between $\sgt{\myp}{\myq'}$ and an arbitrary $\myC\myC$-segment.
The remaining part of the ray $\ray{\myo}{\myq'}$ (that is, the ray $\ray{\myo}{\myq'}$ deprived of the segment $\sgt{\myo}{\myq'}$) is crossing free because, by definition, $\myq'$ is not in the interior of the convex hull of $\myC \cup \{\myq\}$.
    
It only remains to prove that the ray $\ray{\myo}{\myp}$ deprived of the segment $\sgt{\myo}{\myp}$ is also crossing free. For the sake of a contradiction, we assume that this is not the case. If a $\myC\myC$-segment crosses the ray $\ray{\myo}{\myp}$ at a point $\myo_2$, then $\myp$ is in the interior of the triangle $\trgl{\myo_1}{\myo_2}{\myp'}$ (Figure~\ref{fig:2OutsideR_CT1}(c)). If a $\myC\myT$-segment $\sgt{\myp_2}{\myq'}$ crosses the ray $\ray{\myo}{\myp}$, then $\myp$ is in the interior of the triangle $\trgl{\myo_1}{\myp_2}{\myp'}$ (Figure~\ref{fig:2OutsideR_CT1}(c)). Both cases contradicts the fact that $\myp \in \myC$.

\begin{figure}[htb]
    \centering
    \begin{tabularx}{\textwidth}{CCCC}
        \includegraphics[scale=\graphicsScale,page=13]{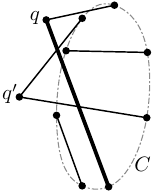} &%
        \includegraphics[scale=\graphicsScale,page=14]{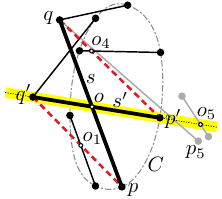} &%
        \includegraphics[scale=\graphicsScale,page=15]{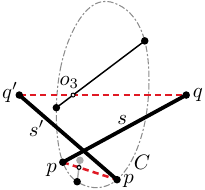} &%
        \includegraphics[scale=\graphicsScale,page=16]{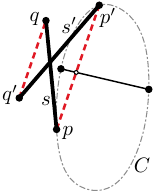}%
        \\
        (a) & (b) & (c) & (d)
    \end{tabularx}
 \caption{Illustrations for Routine~$\myC\myT$ in the proof of Theorem~\ref{thm:2OutsideR}. The existence of the grayed segments is assumed for a contradiction. (a) A $\myC\myT$-segment (in bold) is added back to $\myS$ (in an alternative context compared to Figure~\ref{fig:2OutsideR_CT1}). (b) Illustration for the proof of Lemma~\ref{lem:flipRoutine-CT}; in the case where $\mys'$ is a $\myC\myT$-segment and the inserted segment $\sgt{\myp'}{\myq}$ crosses a $\myC\myC$-segment at $\myo_4$. (c) \& (d) The two cases where the inserted segments are $\sgt{\myq}{\myq'}$ and $\sgt{\myp}{\myp'}$, whether the segment $\sgt{\myq}{\myq'}$ intersects the convex hull of $\myC$.}%
 \label{fig:2OutsideR_CT2}
\end{figure}

\medskip

We now show that the two cases of step~\ref{case:qq'Inserted} are mutually exclusive.

For the segment $\sgt{\myp}{\myp'}$ to have crossings, at least one point of $\myP$ has to be in the interior of the triangle $\trgl{\myp}{\myp'}{\myo}$. 
We examine the following two cases, which are mutually exclusive by definition, and show that the two cases of step~\ref{case:qq'Inserted} respectively imply case~\ref{case:qq'HasCrossings}' and case~\ref{case:pp'HasCrossings}'.

\textbf{Case~\ref{case:qq'HasCrossings}':} In this case, there exists a point $\myo_3$ of the $\myT\myT$-segment $\sgt{\myq}{\myq'}$ that is in the interior of the convex hull of $\myC$ (Figure~\ref{fig:2OutsideR_CT2}(c)). By convexity, the interior of the triangle $\trgl{\myp}{\myp'}{\myo_3}$ does not contain any point of $\myP$.
The triangle $\trgl{\myp}{\myp'}{\myo_3}$ contains the triangle $\trgl{\myp}{\myp'}{\myo}$, and the segments $\sgt{\myo}{\myp}$ and $\sgt{\myo}{\myp'}$ are crossing free. 
Therefore, the $\myC\myC$-segment $\sgt{\myp}{\myp'}$ is crossing free. 

\textbf{Case~\ref{case:pp'HasCrossings}':} In this case, the $\myT\myT$-segment $\sgt{\myq}{\myq'}$ does not intersects the interior of the convex hull of $\myC$ (Figure~\ref{fig:2OutsideR_CT2}(d)). Consequently, the $\myT\myT$-segment $\sgt{\myq}{\myq'}$ is crossing free. In this case, all the crossings of $\myS$ involve the $\myC\myC$-segment $\sgt{\myp}{\myp'}$ and other $\myC\myC$-segments.

\paragraph{Number of flips.} We now analyze the number of flips performed by Routine~$\myC\myT$, which includes the flips in the subroutine Restore-Loop-Invariant.

\begin{lemma}\label{lem:flipRoutine-CT}
    Each call to Routine~$\myC\myT$ performs $ \OO(\myn \degree(\myq') $ flips.
\end{lemma}
\begin{proof}
    We first limit our attention to the sequence $\myQ$ of flips between $\mys$ and $\mys'$ at line~\ref{line:flip} of Routine~$\myC\myT$. There are two types of flips performed, either $\mys'$ is a $\myC\myC$-segment or a $\myC\myT$-segment.
    In the following, we show that:
    \begin{enumerate}
        \item\label{lem:consecutive} The number of consecutive flips in $\myQ$ where $\mys'$ is a $\myC\myC$-segment is at most $\myn-1$.
        \item\label{lem:total} The total number of flips in $\myQ$ where $\mys'$ is a $\myC\myT$-segment is at most $\degree(\myq')$.
    \end{enumerate}
    Lemma~\ref{lem:flipRoutine-CT} follows from~\ref{lem:consecutive} and~\ref{lem:total}, and from the fact that each call to Restore-Loop-Invariant performs $\OO(\myn)$ flips.
    Indeed, the flips performed by Routine~$\myC\myT$ between two flips where $\mys'$ is a $\myC\myT$-segment are either consecutive flips where $\mys'$ is a $\myC\myC$-segment or flips performed by a call to Restore-Loop-Invariant.
    
    To prove~\ref{lem:consecutive}, notice that the proof of Lemma~\ref{lem:farthestFirst} applied to the $\myC\myC$-segments of $\myS$ and $\mys$ ensures that the number of consecutive flips where $\mys'$ is a $\myC\myC$-segment is at most $\myn-1$.

    We now prove~\ref{lem:total}.
    We apply Lemma~\ref{lem:splitting} at the end of a maximal subsequence $\myQ'$ of $\myQ$ consisting of flips where $\mys'$ is a $\myC\myT$-segment. Let $\myk$ be the length of $\myQ'$.
    We first observe that $\myk$ it at most the number of $\myC\myT$-segments crossing $\mys$ just before the first flip of this sequence.
    This observation relies on the correctness of Restore-Loop-Invariant and on the following simple geometric fact.
    Given a triangle $\trgl{\myq}{\myq'}{\myp'}$ and a point $\myo$ on the segment $\sgt{\myp'}{\myq'}$, any segment incident to $\myq'$ and crossing the segment $\sgt{\myp'}{\myq}$ also crosses the segment $\sgt{\myo}{\myq}$.

    We now consider the configuration just after the last flip of $\myQ'$. If $\sgt{\myp'}{\myq}$ is crossing free, then $\myS$ is in fact crossing free and Routine~$\myC\myT$ ends. Otherwise, we have a splitting partition for Lemma~\ref{lem:splitting}, consisting of at most three sets of segments, one for each side of the line $\lineT{\myq'}{\myp'}$, and one for the segments on the line, if any. Next, we show that the line $\lineT{\myq'}{\myp'}$ is crossing free (see Figure~\ref{fig:2OutsideR_CT2}(b) for an illustration of this splitting partition and of the notations introduced in the proof).

    Let $\myo_4$ be the crossing point of $\sgt{\myp'}{\myq}$ and the $\myC\myC$-segment crossing $\sgt{\myp'}{\myq}$ the farthest away from $\myq$.
    For the sake of a contradiction, we assume that the line $\lineT{\myq'}{\myp'}$ is not crossing free. 
    If the line $\lineT{\myq'}{\myp'}$ crosses a $\myC\myC$-segment at a point $\myo_5$, then $\myp'$ is in the interior of the triangle $\trgl{\myo_4}{\myo_5}{\myp}$ (Figure~\ref{fig:2OutsideR_CT2}(b)). 
    If the line $\lineT{\myq'}{\myp'}$ crosses a $\myC\myT$-segment $\sgt{\myq}{\myp_5}$, then $\myp'$ is in the interior of the triangle $\trgl{\myo_4}{\myp_5}{\myp}$ (Figure~\ref{fig:2OutsideR_CT2}(b)). 
    Both cases contradicts the fact that $\myp' \in \myC$. 

    The partition containing the point $\myp$ is crossing free and we claim that it contains at least $\myk$ segments incident to $\myq'$ and it .
    Thus,~\ref{lem:total} follows.
    To prove the claim, 
    it is enough to observe that, at each flip of $\myQ$ where $\mys'$ is a $\myC\myT$-segment, $\mys$ always rotates in the same direction.
    More formally, the sign of the measure between $-\pi$ and $\pi$ of the oriented angle from the vector $\vect{\myp}{\myq}$ to the vector $\vect{\myp'}{\myq}$ is always the same. 
\end{proof}

\subsection{Upper Bound for Two Points Inside or Outside a Convex}
\label{sec:2InsideOutsideR}

The most important results of Theorems~\ref{thm:1Inside1OutsideR}, \ref{thm:2InsideR}, and~\ref{thm:2OutsideR} are summarized in the following theorem.

\begin{theorem}\label{thm:2InsideOutsideR}
Consider a multiset $\myS$ of $\myn$ segments with endpoints $\myP$, such that $\myP$ satisfies the property $\Property$ that $\myP$ is partitioned into $\myP = \myC \cup \myT$ where $\myC$ is in convex position, and $\myT = \{\myq,\myq'\}$.
Let $\myt$ be the sum of the degrees of the points in $\myT$, i.e., $\myt = \degree(\myq) + \degree(\myq')$.
Let $\dR_{\Convex,\Gamma}(\myn)$ be the number of flips to untangle any multiset of at most $\myn$ segments with endpoints in convex position, a graph property $\Gamma$, and removal choice.
There exists a removal strategy $\myR$ such that any untangle sequence of $\myR$ for the graph property $\Gamma$ has length
\[\dR_{\Property,\Gamma}(\myn,\myt) = \OO(\myt^2\myn + \dR_{\Convex,\Gamma}(\myn)).\]

In particular, we have the following upper bounds for different graph properties:
\[\dR_{\Property}(\myn,\myt) = \OO(\myt^2\myn + \myn \log \myn) \text{ and }\]
\[\dR_{\Property,\Cycle}(\myn) , \dR_{\Property,\RedBlue}(\myn) = \OO(\myn) .\]
\end{theorem}

\section{Untangling with Insertion Choice}
\label{cha:I}
In this section, we devise strategies for insertion choice to untangle multisets of segments, therefore providing upper bounds on several versions of $\dI$. Recall that such insertion strategies \emph{do not} choose which pair of crossing segments is removed, but only which pair of segments with the same endpoints is subsequently inserted. We start with a tight bound in $\Convex$ version (where the point set is in convex position), followed by the version where the point set $\myP = \myC \cup \myT$ has some points $\myC$ in convex position and the other points $\myT$ outside the convex hull of $\myC$ and separated from $\myC$ by two parallel lines.

\subsection{Upper Bound for Convex Position}
\label{sec:RUpperConvex}

Let $\myP = \myC= \{\myp_1,\ldots,\myp_{\card{\myC}}\}$ be a set of points in convex position sorted in counterclockwise order along the convex hull boundary (Figure~\ref{fig:convexI}(a)). Given a segment $\myp_\mya\myp_\myb$, we define the \emph{depth} $\PotDepth(\myp_\mya\myp_\myb) = \abs{\myb-\mya}$.\footnote{This definition resembles but is not exactly the same as the depth used in~\cite{BMS19}. It is also similar to the crossing depth defined in Section~\ref{cha:R}.} We use the depth to prove the following theorem.

\begin{figure}[htb]
 \hspace*{\stretch{1}}%
 \pbox[b]{\textwidth}{\centering\includegraphics[scale=\graphicsScale,page=1]{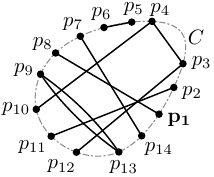}\newline(a)}\hspace*{\stretch{2}}%
 \pbox[b]{\textwidth}{\centering\includegraphics[scale=\graphicsScale,page=2]{convexI}\newline(b)}\hspace*{\stretch{2}}%
 \pbox[b]{\textwidth}{\centering\includegraphics[scale=\graphicsScale,page=3]{convexI}\newline(c)}\hspace*{\stretch{1}}%
 \caption{(a) A multigraph $(\myC,\myS)$ with $\card{\myC}=14$ points in convex position and $\myn=9$ segments. (b) Insertion choice for Case~1 and~2 of the proof of Theorem~\ref{thm:convexI}. (c) Insertion choice for Case~3.}
 \label{fig:convexI}
\end{figure}

\begin{theorem}\label{thm:convexI}
Consider a multiset $\myS$ of $\myn$ segments with endpoints $\myP=\myC$ in convex position.
There exists an insertion strategy $I$ such that any untangle sequence of $I$ has length
\[\dI_{\Convex}(\myn) = \OO(\myn \log \card{\myP}) = \OO(\myn \log \myn).\]
\end{theorem}
\begin{proof}
Let the potential function
\[\PotProduct(\myS) = \prod_{\mys \in \myS} \PotDepth(\mys).\]
As $\PotDepth(\mys) \in \{1,\ldots,\card{\myC}-1\}$, we have that $\PotProduct(\myS)$ is integer, positive, and at most $\card{\myC}^\myn$. Next, we show that for any flipped pair of segments $\myp_\mya\myp_\myb,\myp_\myc\myp_\myd$ there exists an insertion choice that multiplies $\PotProduct(\myS)$ by a factor of at most $3/4$, and the theorem follows.

Consider a flip of a segment $\myp_\mya\myp_\myb$ with a segment $\myp_\myc\myp_\myd$ and assume without loss of generality that $\mya < \myc < \myb < \myd$.
The contribution of the pair of segments $\myp_\mya\myp_\myb,\myp_\myc\myp_\myd$ to the potential $\PotProduct(\myS)$ is the factor $\intFactor=\PotDepth(\myp_\mya\myp_\myb)\PotDepth(\myp_\myc\myp_\myd)$.
Let $\intFactor'$ be the factor corresponding to the pair of inserted segments.

\textbf{Case~1:} If $\PotDepth(\myp_\mya\myp_\myc) \leq \PotDepth(\myp_\myc\myp_\myb)$, then we insert the segments $\myp_\mya\myp_\myc$ and $\myp_\myb\myp_\myd$ and we get $\intFactor'=\PotDepth(\myp_\mya\myp_\myc)\PotDepth(\myp_\myb\myp_\myd)$ (Figure~\ref{fig:convexI}(b)).
We notice $\PotDepth(\myp_\mya\myp_\myb)=\PotDepth(\myp_\mya\myp_\myc)+\PotDepth(\myp_\myc\myp_\myb)$. It follows $\PotDepth(\myp_\mya\myp_\myc) \leq \PotDepth(\myp_\mya\myp_\myb)/2$ and we have $\PotDepth(\myp_\myb\myp_\myd) \leq \PotDepth(\myp_\myc\myp_\myd)$ and then $\intFactor'\leq \intFactor/2$. 

\textbf{Case~2:} If $\PotDepth(\myp_\myb\myp_\myd) \leq \PotDepth(\myp_\myc\myp_\myb)$, then we insert the same segments $\myp_\mya\myp_\myc$ and $\myp_\myb\myp_\myd$ as previously. We have $\PotDepth(\myp_\mya\myp_\myc) \leq \PotDepth(\myp_\mya\myp_\myb)$ and $\PotDepth(\myp_\myb\myp_\myd)\leq \PotDepth(\myp_\myc\myp_\myd)/2$, which gives $\intFactor'\leq \intFactor/2$.

\textbf{Case~3:} If (i) $\PotDepth(\myp_\mya\myp_\myc) > \PotDepth(\myp_\myc\myp_\myb)$ and (ii) $\PotDepth(\myp_\myb\myp_\myd) > \PotDepth(\myp_\myc\myp_\myb)$, then we insert the segments $\myp_\mya\myp_\myd$ and $\myp_\myc\myp_\myb$ (Figure~\ref{fig:convexI}(c)). 
The contribution of the new pair of segments is $\intFactor'=\PotDepth(\myp_\mya\myp_\myd)\PotDepth(\myp_\myc\myp_\myb)$.
We introduce the coefficients $\myx=\frac{\PotDepth(\myp_\mya\myp_\myc)}{\PotDepth(\myp_\myc\myp_\myb)}$ and $\myy=\frac{\PotDepth(\myp_\myb\myp_\myd)}{\PotDepth(\myp_\myc\myp_\myb)}$ so that $\PotDepth(\myp_\mya\myp_\myc) = \myx\PotDepth(\myp_\myc\myp_\myb)$ and $\PotDepth(\myp_\myb\myp_\myd) = \myy\PotDepth(\myp_\myc\myp_\myb)$. It follows that $\PotDepth(\myp_\mya\myp_\myb) = (1+\myx)\PotDepth(\myp_\myc\myp_\myb)$, $\PotDepth(\myp_\myc\myp_\myd)=(1+\myy)\PotDepth(\myp_\myc\myp_\myb)$ and $\PotDepth(\myp_\mya\myp_\myd) = (1+\myx+\myy)\PotDepth(\myp_\myc\myp_\myb)$. The ratio $\intFactor'/\intFactor$ is equal to a function $\myg(\myx,\myy) = \frac{1+\myx+\myy}{(1+\myx)(1+\myy)}$. Due to (i) and (ii), we have that $\myx\geq 1$ and $\myy \geq 1$. 
In other words, we can upper bound the ratio $\intFactor'/\intFactor$ by the maximum of the function $\myg(\myx,\myy)$ with $\myx,\myy \geq 1$. It is easy to show that the function $\myg(\myx,\myy)$ is decreasing with both $\myx$ and $\myy$. Then its maximum 
is obtained for $\myx=\myy=1$ and it is equal to $3/4$, showing that $\intFactor'\leq 3\intFactor/4$.
\end{proof}

\subsection{Upper Bound for Points Separated by Two Parallel Lines}
\label{IUpper2Parallels}

In this section, we prove the following theorem, which is a generalization of Theorem~\ref{thm:convexI}.
We extend our standard general position assumptions to also exclude pairs of endpoints with the same $\myy$-coordinate.

\begin{figure}[htb]
 \hspace*{\stretch{1}}%
 \pbox[b]{\textwidth}{\centering\includegraphics[scale=\graphicsScale,page=1]{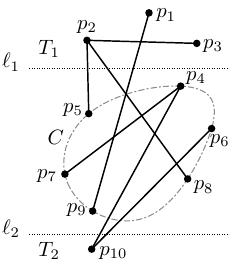}\newline(a)}\hspace*{\stretch{2}}%
 \pbox[b]{\textwidth}{\centering\includegraphics[scale=\graphicsScale,page=2]{separatedI}\newline(b)}\hspace*{\stretch{1}}%
 \caption{(a) Statement of Theorem~\ref{thm:separatedI}. (b) Some insertion choices in the proof of Theorem~\ref{thm:separatedI}.}
 \label{fig:separatedI}
\end{figure}

\begin{theorem}\label{thm:separatedI}
Consider a multiset $\myS$ of $\myn$ segments with endpoints $\myP$, such that $\myP$ satisfies the property $\Property$ that $\myP$ is partitioned into $\myP = \myC \cup \myT_1 \cup \myT_2$ where $\myC$ is in convex position and there exist two horizontal lines $\myl_1,\myl_2$, with $\myT_1$ above $\myl_1$ above $\myC$ above $\myl_2$ above $\myT_2$.
Let $\myt$ be the sum of the degrees of the points in $\myT = \myT_1 \cup \myT_2$.
There exists an insertion strategy $I$ such that any untangle sequence of $I$ has length
\[\dI_{\Property}(\myn,\myt) = \OO(\myt \card{\myP} \log \card{\myC} + \myn \log \card{\myC}) = \OO(\myt\myn \log \myn).\]
\end{theorem}
\begin{proof}
We start by describing the insertion choice for flips involving at least one point in $\myT$.
Let $\myp_1,\ldots,\myp_{\card{\myP}}$ be the points $\myP$ sorted vertically from top to bottom.
Consider a flip involving the points $\myp_\mya,\myp_\myb,\myp_\myc,\myp_\myd$ with $\mya<\myb<\myc<\myd$. The insertion choice is to create the segments $\myp_\mya\myp_\myb$ and $\myp_\myc\myp_\myd$. See Figure~\ref{fig:separatedI}(b). As in~\cite{BoM16} (specifically, in the proof of the upper bound of Theorem~2), we define the potential $\PotParaLine$ of a segment $\myp_\myi\myp_\myj$ as
\[\PotParaLine(\myp_\myi\myp_\myj) = \abs{\myi-\myj}.\]
Notice that $\PotParaLine$ is an integer between $1$ and $\card{\myP}-1$. We define $\PotParaLine_\myT(\myS)$ as the sum of $\PotParaLine(\myp_\myi\myp_\myj)$ for $\myp_\myi\myp_\myj \in \myS$ with $\myp_\myi$ or $\myp_\myj$ in $\myT$. Notice that $0 < \PotParaLine_\myT(\myS) < \myt \card{\myP}$. It is easy to verify that any flip involving a point in $\myT$ decreases $\PotParaLine_\myT(\myS)$ and other flips do not change $\PotParaLine_\myT(\myS)$. Hence, the number of flips involving at least one point in $\myT$ is $\OO(\myt\card{\myP})$.

For the flips involving only points of $\myC$, we use the same choice as in the proof of Theorem~\ref{thm:convexI}.
The potential function 
 \[\PotProduct(\myS) = \prod_{\myp_\myi\myp_\myj \in \myS \;:\; \myp_\myi\in \myC \text{ and } \myp_\myj \in \myC} \PotDepth(\myp_\myi\myp_\myj)\]
is at most $\card{\myC}^\myn$ and decreases by a factor of at most $3/4$ at every flip that involves only points of $\myC$.

However, $\PotProduct(\myS)$ may increase by a factor of $\OO(\card{\myC}^2)$ when performing a flip that involves a point in $\myT$. As such flips only happen $\OO(\myt\card{\myP})$ times, the total increase is at most a factor of $\card{\myC}^{\OO(\myt\card{\myP})}$.

Concluding, the number of flips involving only points in $\myC$ is at most 
\[\log_{4/3}\left(\card{\myC}^{\OO(\myn)} \card{\myC}^{\OO(\myt\card{\myP})} \right) = \OO(\myn \log \card{\myC} + \myt\card{\myP} \log \card{\myC}).\qedhere\]
\end{proof}

\section{Untangling with Both Choices}
\label{cha:RI}
In this section, we devise strategies for insertion and removal choices to untangle a multiset of segments, therefore providing upper bounds on $\dRI$ with some endpoints outside (but not inside) a convex polygon. Recall that such strategies choose which pair of crossing segments is removed \emph{and} which pair of segments with the same endpoints is subsequently inserted.

Throughout this section, we assume that the point set $\myP$ is partitioned into $\myP = \myC \cup \myT$ where $\myC$ is in convex position and the points $\myT$ lie outside the convex hull of $\myC$.
We start with the case where $\myT$ is separated by two parallel lines from $\myC$. Afterwards, we prove an important lemma and apply it to untangle a matching.

\subsection{Upper Bound for Points Separated by Two Parallel Lines}
\label{sec:separated}

In this section, we prove an upper bound when $T$ is separated from $C$ by two parallel lines. In this version, our bound of $\OO(\myn + \myt \card{\myP})$ interpolates the tight convex bound of $\OO(\myn)$ from~\cite{BMS19,FGR24} and the $\OO( \myt \card{\myP})$ bound from~\cite{BoM16} for $\myt$ arbitrary segments.
We extend our standard general position assumptions to also exclude pairs of endpoints with the same $\myy$-coordinate.

\begin{theorem}\label{thm:separatedRI}
Consider a multiset $\myS$ of $\myn$ segments with endpoints $\myP$, such that $\myP$ satisfies the following property $\Property$. The endpoints $\myP$ are partitioned into $\myP = \myC \cup \myT_1 \cup \myT_2$ where $\myC$ is in convex position and there exist two horizontal lines $\myl_1,\myl_2$, with $\myT_1$ above $\myl_1$ above $\myC$ above $\myl_2$ above $\myT_2$.
Let $\myt$ be the sum of the degrees of the points in $\myT= \myT_1 \cup \myT_2$.

There exists a removal and insertion strategy $RI$ such that any untangle sequence of $RI$ has length
\[\dRI_{\Property}(\myn,\myt) = \OO ( \myn + \myt \card{ \myP } ) = \OO ( \myt \myn ).\]
\end{theorem}
\begin{proof}
The algorithm runs in two phases.

\paragraph{Phase~1.} We use removal choice to perform the flips involving a point in $\myT$. At the end of the first phase, there can only be crossings among segments with all endpoints in $\myC$. 
The insertion choice for the first phase is the following.
Let $\myp_1,\ldots,\myp_{\card{\myP}}$ be the points $\myP$ sorted vertically from top to bottom.
Consider a flip involving the points $\myp_\mya,\myp_\myb,\myp_\myc,\myp_\myd$ with $\mya<\myb<\myc<\myd$. The insertion choice is to create the segments $\myp_\mya\myp_\myb$ and $\myp_\myc\myp_\myd$. As in the proofs of Theorem~\ref{thm:convexI} and of Theorem~\ref{thm:separatedI}, we define the potential $\PotDepth$ of a segment $\sgt{\myp_\myi}{\myp_\myj} $ as
$\PotDepth(\myp_\myi\myp_\myj) = \abs{\myi-\myj}$.\footnote{This definition resembles but is not the same as the depth used in~\cite{BMS19}.}
Notice that $\PotDepth$ is an integer from $1$ to $\card{\myP}-1$. We define $\PotDepth(\myS)$ as the sum of $\PotDepth(\myp_\myi\myp_\myj)$ for $\myp_\myi\myp_\myj \in \myS$ with $\myp_\myi$ or $\myp_\myj$ in $\myT$. Notice that $0 < \PotDepth(\myS) < \myt \card{\myP}$. It is easy to verify that any flip involving a point in $\myT$ decreases $\PotDepth(\myS)$. Hence, the number of flips in Phase~1 is $\OO(\myt\card{\myP})$.

\paragraph{Phase~2.}
Since $\myT$ is outside the convex hull of $\myC$, flips between segments with all endpoints in $\myC$ cannot create crossings with the other segments, which are guaranteed to be crossing free at this point. Hence, it suffices to run the algorithm to untangle a convex set with removal and insertion choice from~\cite{BMS19}\footnote{The algorithm in ~\cite{BMS19} is originally proven for convex matchings, but we use here a straightforward generalization to multigraphs available in~\cite{Riv23}, Theorem~3.2.13.}, which performs $\OO(\myn)$ flips.
\end{proof}

\subsection[Upper Bound with Points Outside a Convex]{Upper Bound for Matchings with Points Outside a Convex}
\label{sec:outside}
In this section, we consider the case of endpoints $C \cup T$ with $C$ in convex position and $T$ outside $C$. The result only apply to matching because Lemma~\ref{lem:libLine} is false when multiple copies of a segment are allowed. We start by proving Lemma~\ref{lem:libLine} and then apply it to prove Theorem~\ref{thm:nearConvexRI}.

\subsubsection{Liberating a Line}
\label{sec:libLine}

In this section, we prove the following key lemma, which we use in the following section. The lemma only applies to matchings and it is easy to find a counter-example for multisets ($\myS$ consisting of $\myn$ copies of a single segment that crosses $\myp\myq$).

\begin{lemma}
    \label{lem:libLine}
    Consider a set $\myS$ of $\myn$ segments with endpoints $\myC$ in convex position and a segment $\sgt{\myq}{\myq'}$ intersecting the interior of the convex hull of $\myC$ such that $(\myC \cup \{\myq,\myq'\},\myS \cup \{\sgt{\myq}{\myq'}\})$ forms a matching.
    
    There exists a flip sequence starting at $\myS \cup \{\sgt{\myq}{\myq'}\}$ of length $\OO(\myn)$ which ends with a set of segments that do not cross the line $\lineT{\myq}{\myq'}$ (the line $\lineT{\myq}{\myq'}$ splits the final set of segments). 
\end{lemma}

\begin{proof}
    For each flip performed in the subroutine described hereafter, at least one of the inserted segments does not cross the line $\myq\myq'$ and is removed from $\myS$ (see Figure~\ref{fig:libLine}). 
    
    \begin{figure}[htb]
     \centering
     \includegraphics[scale=\graphicsScale,page=1]{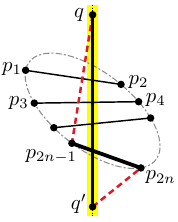} \hfill
     \includegraphics[scale=\graphicsScale,page=2]{libLine} \hfill
     \includegraphics[scale=\graphicsScale,page=3]{libLine} \hfill
     \includegraphics[scale=\graphicsScale,page=4]{libLine} \hfill
     \includegraphics[scale=\graphicsScale,page=5]{libLine}\\
     \caption{An untangle sequence of the subroutine to liberate the line $\myq\myq'$ (with $\myn=4$).}
     \label{fig:libLine}
    \end{figure}
    
    \paragraph{Preprocessing.}
    First, we remove from $\myS$ the segments that do not intersect the line $\myq\myq'$, as they are irrelevant.
    Second, anytime two segments in $\myS$ cross, we flip them choosing to insert the pair of segments not crossing the line $\myq\myq'$. One such flip removes two segments from $\myS$.
    Let $\myp_1\myp_2$ (respectively $\myp_{2 \myn-1}\myp_{2 \myn}$) be the segment in $\myS$ whose intersection point with $\myq\myq'$ is the closest from $\myq$ (respectively $\myq'$).
    Without loss of generality, assume that the points $\myp_1$ and $\myp_{2 \myn-1}$ are on the same side of the line $\myq\myq'$.
    
    \paragraph{First flip.}
    Lemma~\ref{lem:triangleHide} applied to the segment $\myq\myq'$ and the triangle $\myp_1\myp_2\myp_{2 \myn-1}$ shows that at least one of the segments among $\myq\myp_{2 \myn-1},\myq'\myp_1,\myq'\myp_2$ intersects all the segments of $\myS$.
    Without loss of generality, assume that $\myq\myp_{2 \myn-1}$ is such a segment, i.e., that $\myq\myp_{2 \myn-1}$ crosses all segments of $\myS \setminus \{\myp_{2 \myn-1}\myp_{2 \myn}\}$. 
    We choose to remove the segments $\myq\myq'$ and $\myp_{2 \myn-1}\myp_{2 \myn}$, and we choose to insert the segments $\myq\myp_{2 \myn-1}$ and $\myq'\myp_{2 \myn}$.
    As the segment $\myq'\myp_{2 \myn}$ does not cross the line $\myq\myq'$, we remove it from $\myS$.
    
    \paragraph{Second flip.} We choose to flip the segments $\myq\myp_{2 \myn-1}$ and $\myp_1\myp_2$.
    If $\myn$ is odd, we choose to insert the pair of segments $\myq\myp_1,\myp_2\myp_{2 \myn-1}$.
    If $\myn$ is even, we insert the segments $\myq\myp_2,\myp_1\myp_{2 \myn-1}$. 
    
    By convexity, one of the inserted segment (the one with endpoints in $\myC$) crosses all other $\myn-2$ segments.
    The other inserted segment (the one with $\myq$ as one of its endpoints) does not cross the line $\myq\myq'$, so we remove it from $\myS$.
    Note that the condition on the parity of $\myn$ is there only to ensure that the last segment $\myp_{2 \myn-3}\myp_{2 \myn-2}$ is dealt with at the last flip.
    
    \paragraph{Remaining flips.} 
    We describe the third flip. The remaining flips are performed similarly.
    Let $\mys$ be the previously inserted segment.
    Let $\myp_3\myp_4$ be the segment in $\myS$ whose intersection point with $\myq\myq'$ is the closest from $\myq$. Without loss of generality, assume that $\myp_3$ is on the same side of the line $\myq\myq'$ as $\myp_1$ and $\myp_{2 \myn-1}$. 
    
    We choose to flip $\mys$ with $\myp_3\myp_4$.
    If $\mys = \myp_2\myp_{2 \myn-1}$, we choose to insert the pair of segments $\myp_2\myp_4,\myp_3\myp_{2 \myn-1}$.
    If $\mys = \myp_1\myp_{2 \myn-1}$, we choose to insert the pair of segments $\myp_1\myp_3,\myp_4\myp_{2 \myn-1}$.
    
    By convexity, one inserted segment (the one with $\myp_{2 \myn-1}$ as an endpoint) crosses all other $\myn-3$ segments. 
    The other inserted segment does not cross the line $\myq\myq'$, so we remove it from $\myS$.
    Note that the insertion choice described is the only viable one, as the alternative would insert a crossing-free segment crossing the line $\myq\myq'$ that cannot be removed.
\end{proof}

\paragraph{Auxiliary Lemma of Section~\ref{sec:libLine}.}
\label{sec:libLineLemmas}

In this section, we prove Lemma~\ref{lem:triangleHide} used in the proof of Lemma~\ref{lem:libLine}.

Recall that, in the proof of Lemma~\ref{lem:libLine}, we have a convex quadrilateral $\myp_1\myp_2\myp_{2 \myn}\myp_{2 \myn-1}$ and a segment $\myq\myq'$ crossing the segments $\myp_1\myp_2$ and $\myp_{2 \myn}\myp_{2 \myn-1}$ in this order when drawn from $\myq$ to $\myq'$, and we invoke Lemma~\ref{lem:triangleHide} to show that at least one of the segments among $\myq\myp_{2 \myn-1},\myq'\myp_1,\myq'\myp_2$ intersects all the segments of $\myS$.
Before proving Lemma~\ref{lem:triangleHide}, we detail how to apply it to this context.

Lemma~\ref{lem:triangleHide} applied to the segment $\myq\myq'$ and the triangle $\myp_1\myp_2\myp_{2 \myn-1}$ asserts that at least one of the following pairs of segments cross: $\myq\myp_{2 \myn-1},\myp_1\myp_2$, or $\myq'\myp_1,\myp_2\myp_{2 \myn-1}$, or $\myq'\myp_2,\myp_1\myp_{2 \myn-1}$.
If the segments $\myq\myp_{2 \myn-1},\myp_1\myp_2$ cross, then we are done.
If the segments $\myq'\myp_1,\myp_2\myp_{2 \myn-1}$ cross, then the segments $\myq'\myp_1,\myp_{2 \myn}\myp_{2 \myn-1}$ also cross and we are done.
If the segments $\myq'\myp_2,\myp_1\myp_{2 \myn-1}$ cross, then the segments $\myq'\myp_2,\myp_{2 \myn}\myp_{2 \myn-1}$ also cross and we are done.

Next, we state and prove Lemma~\ref{lem:triangleHide}.

\begin{lemma}
    \label{lem:triangleHide}
    For any triangle $\trgl{\myp_1}{\myp_2}{\myp_3}$, for any segment $\sgt{\myq}{\myq'}$ intersecting the interior of the triangle $\trgl{\myp_1}{\myp_2}{\myp_3}$, there exists a segment $\mys \in \{\sgt{\myq}{\myp_1},\sgt{\myq}{\myp_2},\sgt{\myq}{\myp_3},\sgt{\myq'}{\myp_1},\sgt{\myq'}{\myp_2},\sgt{\myq'}{\myp_3}\}$ that intersects the interior of the triangle $\trgl{\myp_1}{\myp_2}{\myp_3}$.
\end{lemma}
\begin{proof}
    If all $\myp_1,\myp_2,\myp_3,\myq,\myq'$ are in convex position, then $\myq$ and the point among $\myp_1,\myp_2,\myp_3$ that is not adjacent to $\myq$ on the convex hull boundary define the segment $\mys$. Otherwise, since $\myq,\myq'$ are not adjacent on the convex hull boundary, assume without loss of generality that $\myp_1$ is not a convex hull vertex and $\myq,\myp_2,\myq',\myp_3$ are the convex hull vertices in order. Then, either the segment $\sgt{\myp_1}{\myq}$ or the segment $\sgt{\myp_1}{\myq'}$ intersects the segment $\sgt{\myp_2}{\myp_3}$.
\end{proof}

\subsubsection{Proof of the Upper Bound}
We are now ready to prove the following theorem, which only applies to matchings because it uses Lemma~\ref{lem:libLine}.

\begin{theorem}\label{thm:nearConvexRI}
Consider a set $\myS$ of $\myn$ segments with endpoints $\myP$ partitioned into $\myP = \myC \cup \myT$ where $\myC$ is in convex position and $\myT$ is outside the convex hull of $\myC$ and such that $(\myP,\myS)$ defines a matching. Let $\myt = |T|$.

There exists a removal and insertion strategy $RI$ such that any untangle sequence of $RI$ has length
\[\dRI_{\Matching}(\myn,\myt) = \OO(\myt^3\myn).\]
\end{theorem}
\begin{proof}
Throughout this proof, we partition the $\myT\myT$-segments into two types: \emph{$\myT\myT\inner$-segment} if it intersects the interior of the convex hull of $\myC$ and \emph{$\myT\myT\outter$-segment} otherwise.
We define the potential $\PotLine_\myl(\myS)$ of a line $\myl$ as the number of segments of $\myS$ crossing $\myl$.

\paragraph{$\myT\myT$-segments.}
At any time during the untangle procedure, if there is a $\myT\myT\inner$-segment $\mys$ that crosses more than $\myt$ segments, we apply Lemma~\ref{lem:libLine} to liberate $\mys$ from every $\myC\myC$-segment using $\OO(\myn)$ flips.
Let $\myl$ be the line containing $\mys$. Since $\PotLine_\myl$ cannot increase (by an easy observation that first appears in~\cite{LS80}), $\PotLine_\myl < \myt$ after Lemma~\ref{lem:libLine}, and there are $\OO(\myt^2)$ different $\myT\myT\inner$-segments in $\myS$, it follows that Lemma~\ref{lem:libLine} is applied $\OO(\myt^2)$ times, performing a total $\OO(\myt^2\myn)$ flips.
As the number of times $\mys$ is inserted and removed differ by at most $1$ and $\PotLine_\myl$ decreases at each flip that removes $\mys$, it follows that $\mys$ participates in $\OO(\myt)$ flips. As there are $\OO(\myt^2)$ different $\myT\myT\inner$-segments in $\myS$, the total number of flips involving $\myT\myT\inner$-segments is $\OO(\myt^3)$.

We define a set $\myL$ of $\OO(\myt)$ lines as follows. For each point $\myq \in \myT$, we have two lines $\myl_1, \myl_2 \in \myL$ that are the two tangents of the convex hull of $\myC$ that pass through $\myq$. As the lines $\myl \in \myL$ do not intersect the interior of the convex hull of $\myC$, the potential $\PotLine_\myl = \OO(\myt)$.
When flipping a $\myT\myT\outter$-segment $\myq_1\myq_2$ with another segment $\myq_3\myp$ with $\myq_3 \in \myT$ ($\myp$ may be in $\myT$ or in $\myC$), we make the insertion choice of creating a $\myT\myT\outter$-segment $\myq_1\myq_3$ such that there exists a line $\myl \in \myL$ whose potential $\PotLine_\myl$ decreases. It is easy to verify that $\myl$ always exist (see Lemma~\ref{lem:Icritical} and Lemma~\ref{lem:criticalTangent}). Hence, the number of flips involving $\myT\myT\outter$-segments is $\OO(\myt^2)$ and the number of flips involving $\myT\myT$-segments in general is $\OO(\myt^3)$.

\paragraph{All except pairs of $\myC\myC$-segments.}
We keep flipping segments that are not both $\myC\myC$-segments with the following insertion choices.
Whenever we flip two $\myC\myT$-segments, we make the insertion choice of creating a $\myT\myT$-segment. Hence, as the number of flips involving $\myT\myT$-segments is $\OO(\myt^3)$, so is the number of flips of two $\myC\myT$-segments.

Whenever we flip a $\myC\myT$-segment $\myp_1\myq$ with $\myq \in \myT$ and a $\myC\myC$-segment $\myp_3\myp_4$, we make the following insertion choice. Let $\myv(\myq)$ be a vector such that the dot product $\myv(\myq) \cdot \myq < \myv(\myq) \cdot \myp$ for all $\myp \in \myC$, that is, $\myv$ is orthogonal to a line $\myl$ separating $\myq$ from $\myC$ and $\myv$ is pointing towards $\myC$. We define the potential $\PotDotprod(\myp_\myi\myq)$ of a segment with $\myp_\myi \in \myC$ and $\myq \in \myT$ as the number of points $\myp \in \myC$ such that $\myv(\myq) \cdot \myp < \myv(\myq) \cdot \myp_\myi$, that is the number of points in $\myC$ before $\myp_\myi$ in direction $\myv$. We choose to insert the segment $\myp_\myi\myq$ that minimizes $\PotDotprod(\myp_\myi\myq)$ for $\myi \in \{3,4\}$. Let $\PotDotprod(\myS)$ be the sum of $\PotDotprod(\myp_\myi\myq)$ for all $\myC\myT$-segments $\myp_\myi\myq$ in $\myS$. It is easy to see that $\PotDotprod(\myS)$ is $\OO(\myt\card{\myC})$ and decreases at each flip involving a $\myC\myT$-segment (not counting the flips inside Lemma~\ref{lem:libLine}).

There are two situation in which $\PotDotprod(\myS)$ may increase. One is when Lemma~\ref{lem:libLine} is applied, which happens $\OO(\myt^2)$ times. Another one is when a $\myT\myT$-segment and a $\myC\myC$-segment flip, creating two $\myC\myT$-segments, which happens $\OO(\myt^3)$ times. At each of these two situations, $\PotDotprod(\myS)$ increases by $\OO(\card{\myC})$. Consequently, the number of flips between a $\myC\myT$-segment and a $\myC\myC$-segment is $\OO(\myt^3\card{\myC}) = \OO(\myt^3\myn)$.

\paragraph{$\myC\myC$-segments.}
By removal choice, we choose to flip the pairs of $\myC\myC$-segments last (except for the ones flipped in Lemma~\ref{lem:libLine}). As $\myT$ is outside the convex hull of $\myC$, flipping two $\myC\myC$-segments does not create crossings with other segments (by Lemma~\ref{lem:splitting}). Hence, we apply Theorem~5 from~\cite{BMS19} to untangle the remaining segments using $\OO(\myn)$ flips.
\end{proof}

\paragraph{Auxiliary Lemmas of Section~\ref{sec:outside}}
\label{sec:outsideLemmas}

In this section, we prove Lemma~\ref{lem:criticalTangent} and Lemma~\ref{lem:Icritical} used in the proof of Theorem~\ref{thm:nearConvexRI}.

Recall that, in the proof of Theorem~\ref{thm:nearConvexRI}, we define a set $\myL$ of lines as follows.
For each point $\myq \in \myT$, we have two lines $\myl_1, \myl_2 \in \myL$ that are the two tangents of the convex hull of $\myC$ that pass through $\myq$.
When flipping a $\myT\myT\outter$-segment $\myq_1\myq_2$ with another segment $\myq_3\myp$ with $\myq_3 \in \myT$ ($\myp$ may be in $\myT$ or in $\myC$), we make the insertion choice of creating a $\myT\myT\outter$-segment $\myq_1\myq_3$ such that there exists a line $\myl \in \myL$ whose potential $\PotLine_\myl$ decreases.
We invoke Lemma~\ref{lem:Icritical} and Lemma~\ref{lem:criticalTangent} to show that such a line $\myl$ always exist.

Indeed, by Lemma~\ref{lem:Icritical}, it is enough to show that there exists a line $\myl \in \myL$ containing one of the points $\myq_1,\myq_2,\myq_3$ that crosses one of the segments $\myq_1\myq_2$ or $\myq_3\myp$. This is precisely what Lemma~\ref{lem:criticalTangent} shows.

Next, we state prove Lemma~\ref{lem:Icritical} and Lemma~\ref{lem:criticalTangent}.

\begin{lemma}
    \label{lem:Icritical}
    Consider two crossing segments $\myp_1\myp_2,\myp_3\myp_4$ and a line $\myl$ containing $\myp_1$ and crossing $\myp_3\myp_4$.
    Then, one of the two pairs of segments $\myp_1\myp_3,\myp_2\myp_4$ or $\myp_1\myp_4,\myp_2\myp_3$ does not cross $\myl$.
    In other words, there exists an insertion choice for a flip removing $\myp_1\myp_2,\myp_3\myp_4$ such that the number of segments crossing $\myl$ decreases.
\end{lemma}
\begin{proof}
    Straightforward.
\end{proof}

\begin{lemma}
    \label{lem:criticalTangent}
    Consider a closed convex body $\myB$ and two crossing segments $\myq_1\myq_3,\myq_2\myq_4$ whose endpoints $\myq_1,\myq_2,\myq_3$ are not in $\myB$, and whose endpoint $\myq_4$ is not in the interior of $\myB$.
    If the segment $\myq_1\myq_3$ does not intersect the interior of $\myB$, then at least one of the six lines tangent to $\myB$ and containing one of the endpoints $\myq_1,\myq_2,\myq_3$ is crossing one of the segments $\myq_1\myq_3,\myq_2\myq_4$ (\Figure~\ref{fig:criticalTangent}(a)).
    (General position is assumed, meaning that the aforementioned six lines are distinct, i.e., each line does not contain two of the points $\myq_1,\myq_2,\myq_3,\myq_4$.)
\end{lemma}

\begin{figure}[!ht]
    \hspace*{\stretch{1}}%
    \pbox[b]{\textwidth}{\centering\includegraphics[page=1,scale=\graphicsScale]{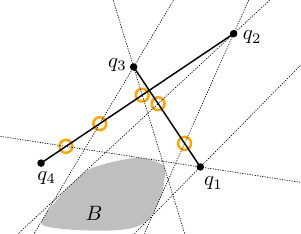}\newline(a)}\hspace*{\stretch{2}}%
    \pbox[b]{\textwidth}{\centering\includegraphics[page=2,scale=\graphicsScale]{criticalTangent}\newline(b)}%
    \hspace*{\stretch{1}}%
    \caption{(a) In the statement of Lemma~\ref{lem:criticalTangent}, we assert the existence of points, circled in the figure, which are the intersection of a line tangent to $\myB$ and containing one of the points $\myq_1,\myq_2,\myq_3$. (b) In the proof of Lemma~\ref{lem:criticalTangent} by contraposition, we exhibit a point, circled in the figure, showing that $\myB$ intersects one of the segment $\myq_1\myq_3$.}
    \label{fig:criticalTangent}%
\end{figure}

\begin{proof}
    For all $i \in \{1,2,3\}$, let $\myl_\myi$ and $\myl_\myi'$ be the two lines containing $\myq_\myi$ and tangent to $\myB$ (\Figure~\ref{fig:criticalTangent}(b)).
    By contraposition, we assume that none of the six lines $\myl_1,\myl_1',\myl_2,\myl_2',\myl_3,\myl_3'$ crosses one of the segments $\myq_1\myq_3,\myq_2\myq_4$. In other words, we assume that the six lines are tangent to the convex quadrilateral $\myq_1\myq_2\myq_3\myq_4$.
    It is well known that, if $\mym \geq 5$, then any arrangement of $\mym$ lines or more admits at most one face with $\mym$ edges (see~\cite{Gru73} for example).
    Therefore, $\myB$ is contained in the same face of the arrangement of the six lines as the quadrilateral $\myq_1\myq_2\myq_3\myq_4$.
    Let $\ppx_1$ (respectively $\ppx_1'$) be a contact point between the line $\myl_1$ (respectively $\myl_1'$) and the convex body $\myB$.
    The segment $\ppx_1\ppx_1'$ crosses the segment $\myq_1\myq_3$ and is contained in $\myB$ by convexity  (\Figure~\ref{fig:criticalTangent}(b)), concluding the proof by contraposition.
\end{proof}

\section{Conclusion and Open Problems}
\label{cha:conclusion}

Flip graphs of geometric reconfiguration problems are mathematically challenging objects. The flip operation that removes two crossing segments and inserts two non-crossing segments with the same endpoints has applications to local search heuristics and produces directed flip graphs. Such flip graphs have previously been studied in the contexts of TSP tours~\cite{VLe81,OdW07,WCL09}, matchings~\cite{BoM16,BMS19,DDFGR22}, and trees~\cite{BMS19}.

In this paper, we have considered several different versions of the problem within the unifying framework of removal and insertion choices.
In all versions, we consider a multiset of $\myn$ segments.
When the endpoints are in convex position, we showed that removal choice is enough to get an $\OO(\myn \log \myn)$ bound on the number of flips, implying the same bound for trees (matching the bound in~\cite{BMS19}) but without using any specific property of trees. It is an open problem whether any of these two bounds (the general removal choice bound or the special tree bound\footnote{A fatal flaw has been found in the proof of the $\OO(\myn)$ bound for trees presented in the PhD dissertation~\cite{Riv23}.}) could be improved to $\OO(\myn)$ as in the special cases of TSP tours and red-blue matchings. Still in the convex case, we proved an $\OO(\myn \log \myn)$ bound on the number of flips using only insertion choice. We do not know if similar bounds hold in the simple cases of one endpoint inside the convex hull or two points outside the convex hull that cannot be separated by two parallel lines.

The remaining results consider near-convex sets of endpoints. The most general of these results is in the case of no choice, where the endpoints $T$ that are not in convex position may come in any number and may be placed anywhere. In contrast, in our removal strategies, at most two endpoints $T$ that are not in convex position are allowed and the proofs are very technical, considering each placement of $T$ separately. In our insertion strategies, the endpoints $T$ must be outside and separated by two parallel lines. In our removal and insertion strategy the points may be anywhere outside the convex hull but multiple copies of the same segment are not allowed anywhere in the flip graph, which limits its application to matchings or some artificial problems such as all but one vertex having degree $1$. Improving and generalizing these bounds remains an open problem.

Our framework with insertion choice is limited to problems where both insertion choices are always available. It is motivating to consider problems where two insertion choices are sometimes available. For example, we may consider the case of maintaining a collection of simple vertex-disjoint graph cycles (without multi-edges). In this case, we have two insertion choices in all but some very specific situations where one of the insertions would create an edge of multiplicity two (that is, a multigraph cycle of length $2$). Obtaining a subcubic bound for endpoints in general position with this limited insertion choice is an open problem which may be easier than proving a subcubic bound when insertion choice is never available (with or without removal choice).

\begin{figure}[!ht]
    \centering%
    \hspace*{\stretch{1}}%
    \includegraphics[scale=\graphicsScale,page=2]{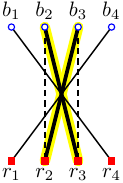}\hspace*{\stretch{1}}%
    \makebox[0pt]{\includegraphics[scale=\graphicsScale,page=14]{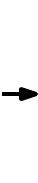}}\hspace*{\stretch{1}}%
    \includegraphics[scale=\graphicsScale,page=4]{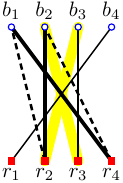}\hspace*{\stretch{1}}%
    \makebox[0pt]{\includegraphics[scale=\graphicsScale,page=14]{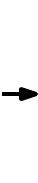}}\hspace*{\stretch{1}}%
    \includegraphics[scale=\graphicsScale,page=6]{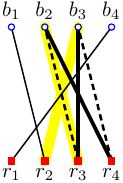}\hspace*{\stretch{1}}%
    \makebox[0pt]{\includegraphics[scale=\graphicsScale,page=14]{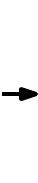}}\hspace*{\stretch{1}}%
    \includegraphics[scale=\graphicsScale,page=8]{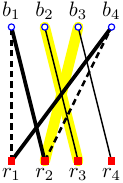}\hspace*{\stretch{1}}%
    \makebox[0pt]{\includegraphics[scale=\graphicsScale,page=14]{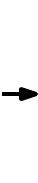}}\hspace*{\stretch{1}}%
    \includegraphics[scale=\graphicsScale,page=10]{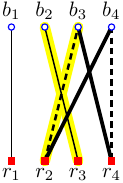}\hspace*{\stretch{1}}%
    \makebox[0pt]{\includegraphics[scale=\graphicsScale,page=14]{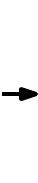}}\hspace*{\stretch{1}}%
    \includegraphics[scale=\graphicsScale,page=12]{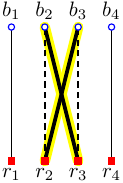}\hspace*{\stretch{1}}%
    \caption{A flip sequence where the highlighted segments $\sgt{\myr_2}{\myb_3},\sgt{\myr_3}{\myb_2}$ are removed than reinserted and flipped again. It is possible to iterate the use of this same flip $\frac{\myn}{2}$ times. The first three flips of this figure appear in~\cite{BoM16}.}%
    \label{fig:flipReuse}%
\end{figure}

Without insertion choice, reducing the gaps between the $\OO(\myn^3)$ upper bound and the lower bounds of $\Omega(\myn^2)$ and $\Omega(\myn)$ (respectively without and with removal choice) remains the most challenging open problem. The $\OO(\myn^{8/3})$ bound on the number of distinct flips is a hopeful step in this direction. 
Note that there exist flip sequences using the same flip a linear number of times (Figure~\ref{fig:flipReuse}).
However, we were not able to find a set of line segments that requires the same pair of segments to be flipped twice in order to be untangled. A proof of this statement would imply a subcubic upper bound with removal choice.

\bibliographystyle{plainurl}
\bibliography{ref}

\begin{thebibliography}{10}

\bibitem{ABDK22}
Oswin Aichholzer, Brad Ballinger, Therese Biedl, Mirela Damian, Erik~D Demaine, Matias Korman, Anna Lubiw, Jayson Lynch, Josef Tkadlec, and Yushi Uno.
\newblock Reconfiguration of non-crossing spanning trees.
\newblock {\em arXiv preprint}, 2022.
\newblock URL: \url{https://arxiv.org/abs/2206.03879}.

\bibitem{ABPS24}
Oswin Aichholzer, Anna Br{\"o}tzner, Daniel Perz, and Patrick Schnider.
\newblock Flips in odd matchings.
\newblock In {\em 40th European Workshop on Computational Geometry (EuroCG)}, volume~40, pages 447--452, 2024.
\newblock URL: \url{https://eurocg2024.math.uoi.gr/data/uploads/paper_59.pdf}.

\bibitem{AKLM23}
Oswin Aichholzer, Kristin Knorr, Maarten L{\"o}ffler, Zuzana Mas{\'a}rov{\'a}, Wolfgang Mulzer, Johannes Obenaus, Rosna Paul, and Birgit Vogtenhuber.
\newblock Flipping plane spanning paths.
\newblock In {\em International Conference and Workshops on Algorithms and Computation (WALCOM)}, 2023.
\newblock URL: \url{https://arxiv.org/abs/2202.10831}, \href {https://doi.org/10.1007/978-3-031-27051-2_5} {\path{doi:10.1007/978-3-031-27051-2_5}}.

\bibitem{AMP15}
Oswin Aichholzer, Wolfgang Mulzer, and Alexander Pilz.
\newblock Flip distance between triangulations of a simple polygon is {NP}-complete.
\newblock {\em Discrete \& Computational Geometry}, 54(2):368--389, 2015.
\newblock URL: \url{https://arxiv.org/abs/1209.0579}, \href {https://doi.org/10.1007/s00454-015-9709-7} {\path{doi:10.1007/s00454-015-9709-7}}.

\bibitem{AIM07}
Selim~G. Akl, Md~Kamrul Islam, and Henk Meijer.
\newblock On planar path transformation.
\newblock {\em Information processing letters}, 104(2):59--64, 2007.
\newblock \href {https://doi.org/10.1016/j.ipl.2007.05.009} {\path{doi:10.1016/j.ipl.2007.05.009}}.

\bibitem{ABCC11}
David~L. Applegate, Robert~E. Bixby, Va{\v{s}}ek Chv{\'a}tal, and William~J. Cook.
\newblock The traveling salesman problem.
\newblock In {\em The Traveling Salesman Problem}. Princeton university press, 2011.

\bibitem{Aro96}
Sanjeev Arora.
\newblock Polynomial time approximation schemes for {Euclidean} {TSP} and other geometric problems.
\newblock In {\em 37th Conference on Foundations of Computer Science}, pages 2--11, 1996.

\bibitem{BeI08}
Sergey Bereg and Hiro Ito.
\newblock Transforming graphs with the same degree sequence.
\newblock In {\em Computational Geometry and Graph Theory}, pages 25--32, 2008.
\newblock \href {https://doi.org/10.1007/978-3-540-89550-3_3} {\path{doi:10.1007/978-3-540-89550-3_3}}.

\bibitem{BeI17}
Sergey Bereg and Hiro Ito.
\newblock Transforming graphs with the same graphic sequence.
\newblock {\em Journal of Information Processing}, 25:627--633, 2017.
\newblock \href {https://doi.org/10.2197/ipsjjip.25.627} {\path{doi:10.2197/ipsjjip.25.627}}.

\bibitem{BMS19}
Ahmad Biniaz, Anil Maheshwari, and Michiel Smid.
\newblock Flip distance to some plane configurations.
\newblock {\em Computational Geometry}, 81:12--21, 2019.
\newblock URL: \url{https://arxiv.org/abs/1905.00791}.

\bibitem{BBH19}
Marthe Bonamy, Nicolas Bousquet, Marc Heinrich, Takehiro Ito, Yusuke Kobayashi, Arnaud Mary, Moritz M{\"{u}}hlenthaler, and Kunihiro Wasa.
\newblock The perfect matching reconfiguration problem.
\newblock In {\em 44th International Symposium on Mathematical Foundations of Computer Science}, volume 138 of {\em LIPIcs}, pages 80:1--80:14, 2019.
\newblock \href {https://doi.org/10.4230/LIPIcs.MFCS.2019.80} {\path{doi:10.4230/LIPIcs.MFCS.2019.80}}.

\bibitem{BoM16}
{\'{E}}douard Bonnet and Tillmann Miltzow.
\newblock Flip distance to a non-crossing perfect matching.
\newblock {\em arXiv}, 1601.05989, 2016.
\newblock URL: \url{http://arxiv.org/abs/1601.05989}.

\bibitem{BFR24}
Valentino Boucard, Guilherme~D. da~Fonseca, and Bastien Rivier.
\newblock Further connectivity results on plane spanning path reconfiguration.
\newblock {\em arXiv}, 2407.00244, 2024.
\newblock URL: \url{https://arxiv.org/abs/2407.00244}.

\bibitem{BJ20}
Nicolas Bousquet and Alice Joffard.
\newblock Approximating shortest connected graph transformation for trees.
\newblock In {\em Theory and Practice of Computer Science}, pages 76--87, 2020.
\newblock \href {https://doi.org/10.1007/978-3-030-38919-2_7} {\path{doi:10.1007/978-3-030-38919-2_7}}.

\bibitem{BuKi22}
Maike Buchin and Bernhard Kilgus.
\newblock Fr{\'e}chet distance between two point sets.
\newblock {\em Computational Geometry}, 102:101842, 2022.
\newblock \href {https://doi.org/10.1016/j.comgeo.2021.101842} {\path{doi:10.1016/j.comgeo.2021.101842}}.

\bibitem{ChWu09}
Jou-Ming Chang and Ro-Yu Wu.
\newblock On the diameter of geometric path graphs of points in convex position.
\newblock {\em Information processing letters}, 109(8):409--413, 2009.
\newblock \href {https://doi.org/10.1016/j.ipl.2008.12.017} {\path{doi:10.1016/j.ipl.2008.12.017}}.

\bibitem{ChC19}
Miroslav Chleb{\'\i}k and Janka Chleb{\'\i}kov{\'a}.
\newblock Approximation hardness of {Travelling} {Salesman} via weighted amplifiers.
\newblock In {\em 25th International Computing and Combinatorics Conference}, pages 115--127, 2019.

\bibitem{Christofides76}
Nicos Christofides.
\newblock Worst-case analysis of a new heuristic for the travelling salesman problem.
\newblock Technical report, Carnegie-Mellon Univ Pittsburgh Pa Management Sciences Research Group, 1976.

\bibitem{FGR23}
Guilherme~D. da~Fonseca, Yan Gerard, and Bastien Rivier.
\newblock On the longest flip sequence to untangle segments in the plane.
\newblock In {\em International Conference and Workshops on Algorithms and Computation (WALCOM 2023)}, volume 13973 of {\em Lecture Notes in Computer Science}, pages 102--112, 2023.
\newblock URL: \url{https://arxiv.org/abs/2210.12036}, \href {https://doi.org/10.1007/978-3-031-27051-2_10} {\path{doi:10.1007/978-3-031-27051-2_10}}.

\bibitem{FGR24}
Guilherme~D. da~Fonseca, Yan Gerard, and Bastien Rivier.
\newblock Short flip sequences to untangle segments in the plane.
\newblock In {\em International Conference and Workshops on Algorithms and Computation (WALCOM 2024)}, volume 14549 of {\em Lecture Notes in Computer Science}, pages 163--178, 2024.
\newblock URL: \url{https://arxiv.org/abs/2307.00853}, \href {https://arxiv.org/abs/2307.00853} {\path{arXiv:2307.00853}}.

\bibitem{DDFGR22}
Arun~Kumar Das, Sandip Das, Guilherme~D. da~Fonseca, Yan Gerard, and Bastien Rivier.
\newblock Complexity results on untangling red-blue matchings.
\newblock {\em Computational Geometry}, 111:101974, 2023.
\newblock URL: \url{https://arxiv.org/abs/2202.11857}, \href {https://doi.org/10.1016/j.comgeo.2022.101974} {\path{doi:10.1016/j.comgeo.2022.101974}}.

\bibitem{Dav10}
Donald Davendra.
\newblock {\em Traveling salesman problem: Theory and applications}.
\newblock BoD--Books on Demand, 2010.

\bibitem{ERV14}
Matthias Englert, Heiko R{\"o}glin, and Berthold V{\"o}cking.
\newblock Worst case and probabilistic analysis of the {2-Opt} algorithm for the {TSP}.
\newblock {\em Algorithmica}, 68(1):190--264, 2014.
\newblock URL: \url{https://link.springer.com/content/pdf/10.1007/s00453-013-9801-4.pdf}.

\bibitem{ELSS73}
Paul Erd{\"o}s, L{\'a}szl{\'o} Lov{\'a}sz, A.~Simmons, and Ernst~G. Straus.
\newblock Dissection graphs of planar point sets.
\newblock In {\em A survey of combinatorial theory}, pages 139--149. Elsevier, 1973.

\bibitem{EKM13}
P{\'e}ter~L. Erd{\H{o}}s, Zolt{\'a}n Kir{\'a}ly, and Istv{\'a}n Mikl{\'o}s.
\newblock On the swap-distances of different realizations of a graphical degree sequence.
\newblock {\em Combinatorics, Probability and Computing}, 22(3):366--383, 2013.
\newblock URL: \url{https://arxiv.org/abs/1205.2842}.

\bibitem{Gru73}
Branko Grünbaum.
\newblock Polygons in arrangements generated by n points.
\newblock {\em Mathematics Magazine}, 46(3):113--119, 1973.
\newblock \href {https://doi.org/10.1080/0025570X.1973.11976293} {\path{doi:10.1080/0025570X.1973.11976293}}.

\bibitem{GuPu06}
Gregory Gutin and Abraham~P. Punnen.
\newblock {\em The traveling salesman problem and its variations}, volume~12.
\newblock Springer Science \& Business Media, 2006.

\bibitem{Hak62}
Seifollah~Louis Hakimi.
\newblock On realizability of a set of integers as degrees of the vertices of a linear graph. {I}.
\newblock {\em Journal of the Society for Industrial and Applied Mathematics}, 10(3):496--506, 1962.

\bibitem{Hak63}
Seifollah~Louis Hakimi.
\newblock On realizability of a set of integers as degrees of the vertices of a linear graph {II}. uniqueness.
\newblock {\em Journal of the Society for Industrial and Applied Mathematics}, 11(1):135--147, 1963.

\bibitem{HNU99}
Ferran Hurtado, Marc Noy, and Jorge Urrutia.
\newblock Flipping edges in triangulations.
\newblock {\em Discrete \& Computational Geometry}, 22(3):333--346, 1999.

\bibitem{phdJof}
Alice Joffard.
\newblock {\em Graph domination and reconfiguration problems}.
\newblock PhD thesis, Université Claude Bernard Lyon 1, 2020.

\bibitem{KKR24}
Linda Kleist, Peter Kramer, and Christian Rieck.
\newblock On the connectivity of the flip graph of plane spanning paths.
\newblock In {\em Graph-Theoretic Concepts in Computer Science (WG 2024)}, 2024.
\newblock URL: \url{https://arxiv.org/abs/2407.03912}.

\bibitem{Law72}
Charles~L. Lawson.
\newblock Transforming triangulations.
\newblock {\em Discrete Mathematics}, 3(4):365--372, 1972.

\bibitem{LuP15}
Anna Lubiw and Vinayak Pathak.
\newblock Flip distance between two triangulations of a point set is {NP}-complete.
\newblock {\em Computational Geometry}, 49:17--23, 2015.
\newblock URL: \url{https://arxiv.org/abs/1205.2425}.

\bibitem{Mit99}
Joseph S.~B. Mitchell.
\newblock Guillotine subdivisions approximate polygonal subdivisions: A simple polynomial-time approximation scheme for geometric {TSP}, k-{MST}, and related problems.
\newblock {\em SIAM Journal on computing}, 28(4):1298--1309, 1999.

\bibitem{NiN18}
Naomi Nishimura.
\newblock Introduction to reconfiguration.
\newblock {\em Algorithms}, 11(4), 2018.
\newblock \href {https://doi.org/10.3390/a11040052} {\path{doi:10.3390/a11040052}}.

\bibitem{OdW07}
Yoshiaki Oda and Mamoru Watanabe.
\newblock The number of flips required to obtain non-crossing convex cycles.
\newblock In {\em Kyoto International Conference on Computational Geometry and Graph Theory}, pages 155--165, 2007.
\newblock \href {https://doi.org/10.1007/978-3-540-89550-3_17} {\path{doi:10.1007/978-3-540-89550-3_17}}.

\bibitem{Pil14}
Alexander Pilz.
\newblock Flip distance between triangulations of a planar point set is {APX}-hard.
\newblock {\em Computational Geometry}, 47(5):589--604, 2014.
\newblock URL: \url{https://arxiv.org/abs/1206.3179}.

\bibitem{RaSm98}
Satish~B. Rao and Warren~D. Smith.
\newblock Approximating geometrical graphs via “spanners” and “banyans”.
\newblock In {\em Proceedings of the thirtieth annual ACM symposium on Theory of computing}, pages 540--550, 1998.

\bibitem{Riv23}
Bastien Rivier.
\newblock {\em {Untangling Segments in the Plane}}.
\newblock {PhD} dissertation, {Universit{\'e} Clermont Auvergne, Clermont-Ferrand, France.}, November 2023.
\newblock URL: \url{https://theses.hal.science/tel-04507187}.

\bibitem{Heu13}
Jan van~den Heuvel.
\newblock The complexity of change.
\newblock {\em Surveys in Combinatorics}, 409:127--160, 2013.

\bibitem{LS80}
Jan van Leeuwen and Anneke~A. Schoone.
\newblock Untangling a traveling salesman tour in the plane.
\newblock Technical Report RUU-CS-80-11, University of Utrecht, 1980.
\newblock URL: \url{https://dspace.library.uu.nl/bitstream/handle/1874/15912/leeuwen_80_untangling_a_traveling.pdf?sequence=1}.

\bibitem{VLe81}
Jan van Leeuwen and Anneke~A. Schoone.
\newblock Untangling a traveling salesman tour in the plane.
\newblock In {\em 7th Workshop on Graph-Theoretic Concepts in Computer Science}, 1981.

\bibitem{BKU24}
H\aa vard Bakke~Bjerkevik, Linda Kleist, Torsten Ueckerdt, and Birgit Vogtenhuber.
\newblock Flipping non-crossing spanning trees.
\newblock In {\em 36th ACM-SIAM Symposium on Discrete Algorithms (SODA)}, 2025.
\newblock URL: \url{https://arxiv.org/abs/2410.23809}.

\bibitem{Wil99}
Todd~G. Will.
\newblock Switching distance between graphs with the same degrees.
\newblock {\em SIAM Journal on Discrete Mathematics}, 12(3):298--306, 1999.

\bibitem{WCL09}
Ro{-}Yu Wu, Jou{-}Ming Chang, and Jia{-}Huei Lin.
\newblock On the maximum switching number to obtain non-crossing convex cycles.
\newblock In {\em 26th Workshop on Combinatorial Mathematics and Computation Theory}, pages 266--273, 2009.

\end{thebibliography}
\end{document}